\documentclass[aps,pra,superscriptaddress,twocolumn]{revtex4-2}

\usepackage{physics}
\usepackage{xcolor}
\usepackage{amssymb}
\usepackage{mathtools}
\usepackage{bbm}
\usepackage{graphicx} 
\usepackage{amsmath}
\usepackage{amsfonts}
\usepackage{braket}
\usepackage{fontenc}
\usepackage{mathrsfs}
\usepackage{verbatim}
\usepackage{dsfont}
\usepackage{csquotes}
\usepackage{comment}
\usepackage{amsthm}
\usepackage[colorlinks,citecolor=blue,urlcolor=blue,bookmarks=false,hypertexnames=true]{hyperref}
\usepackage{setspace}
\usepackage{moresize}
\usepackage{ulem}
\usepackage{bm}
\usepackage{enumitem}

\theoremstyle{plain}
\newtheorem{thm}{Theorem}[]

\newtheorem*{result}{Main Result}
\newtheorem{rslt}{Result}

\newtheorem*{cor*}{Corollary}

\newtheorem*{thm*}{Theorem}
\newtheorem{lem}[thm]{Lemma}

\newtheorem{cor}[thm]{Corollary}
\newtheorem{defn}[thm]{Definition}
\theoremstyle{remark}

\newcommand{\cH}{\mathcal{H}}
\newcommand{\cV}{\mathcal{V}}
\newcommand{\cU}{\mathcal{U}}
\newcommand{\tcV}{\widetilde{\mathcal{V}}}
\newcommand{\tnu}{\widetilde{\nu}}
\newcommand{\cE}{\mathcal{E}}
\def\id{{\mathbb I}}

\newcommand{\sumab}[2]{\underset{#1}{\overset{#2}{\sum}}}
\newcommand{\suma}[1]{\underset{#1}{\sum}}

\newcommand{\rec}[2]{C(#1,#2)}

\newcommand{\aut}{{\tt Eig}}

\newcommand{\ii}{\mathrm{i}}

\begin{document} 

\title{The entropic coherence is a necessary resource for non energy preserving gates}

\date{\today}

\author{Riccardo Castellano}
\email{riccardo.castellano@unige.ch}
\affiliation{Scuola Normale Superiore, 56126 Pisa, Italy}
\affiliation{Dipartimento di Fisica dell’Universit\`a di Pisa, Largo Pontecorvo 3, I-56127 Pisa, Italy}
\affiliation{Department of Applied Physics, University of Geneva, 1211 Geneva, Switzerland}

\author{Vasco Cavina}
\affiliation{Scuola Normale Superiore, 56126 Pisa, Italy}

\author{Martí Perarnau-Llobet}
\affiliation{Department of Applied Physics, University of Geneva, 1211 Geneva, Switzerland}
\affiliation{F\'isica Te\`orica: Informaci\'o i Fen\`omens Qu\`antics, Department de F\'isica, Universitat Aut\`onoma de Barcelona, 08193 Bellaterra (Barcelona), Spain}

\author{Vittorio Giovannetti}
\affiliation{Scuola Normale Superiore, 56126 Pisa, Italy}
\affiliation{NEST, and Istituto Nanoscienze-CNR, 56126 Pisa, Italy}

\author{Pavel Sekatski}
\email{pavel.sekatski@unige.ch}
\affiliation{Department of Applied Physics, University of Geneva, 1211 Geneva, Switzerland}

\begin{abstract}
 We consider the task of implementing non energy preserving gates (NEPG) on a finite-dimensional system $S$ via an energy-preserving interaction with an external battery $B$. We prove that the entropic coherence of the battery (an instance of the relative entropy of resource) is a necessary resource for this task, and find a lower bound on its minimum amount that has to be present in the battery to be able to implement NEPGs with a fixed desired precision. An immediate corollary is that any finite-dimensional battery is doomed to a certain minimal error in the gate implementation task. Moreover, under assumptions on the density of energy levels in the battery Hamiltonian, our main results imply additional lower bounds on the minimal amount of energy and quantum Fisher information required to implement any gate. We show that these bounds can be stronger than the universal bounds previously established in the literature.
\end{abstract}

\maketitle

Understanding the physical requirements for precise control of quantum systems is a central challenge in quantum information science.  Theoretical work has explored the limits of control in a variety of tasks, including quantum measurements~\cite{WAY2006, EnergyMesuramentNavascues, WAY2008, WAYUnbounded(Tajima), WAY2011}, quantum channels~\cite{ChannelsCost(Tajima)}, and state preparation~\cite{AsymmetryPureStates, IIDNonAbelian, IID-Rates}. In parallel, the development of quantum resource theories~\cite{QResourceTh, RECForAsymmetrySpekkens(2009), HolevoAsymmDef(2014)} has provided a framework for analyzing these limitations at a fundamental level.

Among the different control operations, the realisation of precise unitary gates stands as a crucial task, particularly in quantum computation.  Assuming energy conservation as a fundamental symmetry, non-energy-preserving gates (NEPGs) 
on a system $S$ can only be implemented with the aid of an auxiliary system $B$, commonly referred to as a battery. The battery is initialized in a certain state, whose performance is measured 
in terms of how closely the open dynamics of $S$ (induced by the coupling with $B$) resembles the target gate.
Progress has been made in identifying the necessary conditions that the initial preparation of the battery must satisfy to achieve a certain precision. In particular, the Quantum Fisher Information (QFI) with respect to the Hamiltonian and the average energy of the battery are both known to be essential resources~\cite{Chiribella, Chiribella2, WeakQFIbound, QFIBound,SaturatingWAY(Gaussian)}.

In this paper, we show that  the entropic coherence (EC, in Eq.~\eqref{eq: ECE}) of the battery is an essential resource for the implementation of any NEPG. In contrast to the resources mentioned above, this is an entropic (dimensionless) quantity that remains invariant when rescaling the battery Hamiltonian.
This is particularly relevant as it provides an example of a dimensionless resource arising from the requirement of energy conservation, with striking consequences. In particular, our results imply that an ideal implementation of non-energy-preserving gates requires a battery of unbounded dimension.
For a single system the EC is tightly related to the relative entropy of superposition \cite{aberg2006quantifying} and of coherence~\cite{BaumgratzREC}, but behaves differently under system composition and can be seen as a specific instance of the relative entropy of resource defined in~\cite{RECAsUniversalConvertibilityRate(2002), RECForAsymmetrySpekkens(2009),Corrected-Uni-Conv-Rate(2012)}.

{\it Framework.---}
Given a system $S$ with Hamiltonian ${H}_S$, our goal is to implement a generic unitary gate $\mathcal{V}_{S}({\rho}_S) = {V}_S {\rho}_S {V}_S^{\dag}$ using energy-preserving operations. If $S$ is isolated, this requirement restricts us to energy-preserving gates such that $[{V}_S, {H}_S] = 0$. To overcome this limitation, we allow the presence of a battery system $B$ with Hamiltonian ${H}_B$ and prepared in the state $\beta_B$, such that only the total energy $H_{SB}=H_S+H_B$ must be conserved, see Fig.~\ref{Def:MainChannel}.

\begin{figure}
    \includegraphics[width=0.85\linewidth]{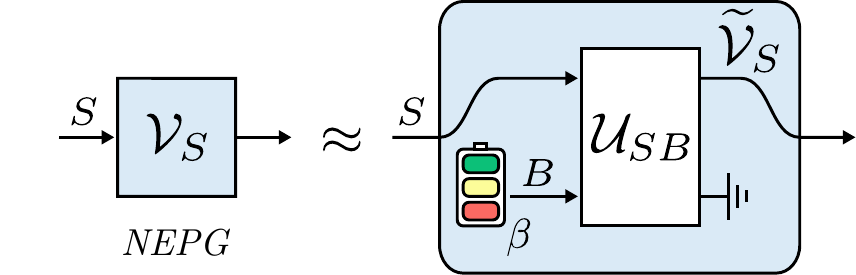}

    \caption{A non-energy-preserving gate (NEPG) $\mathcal{V}_{S}$ on the system $S$ can be approximated with a channel $\widetilde{\mathcal{V}}_{S}$, realized via a joint 
    total-energy-preserving unitary transformation $\mathcal{U}_{SB}$ on the system and a battery $B$. 
    To be useful for the task the initial state of the battery $\beta_B$ must present some coherence.  We relate its entropic coherence $\rec{\beta_B}{H_B}$ in Eq.~\eqref{eq: ECE}, to the error in gate approximation, quantified by the worst-case infidelity $\epsilon_{wc}(\mathcal{V}_{S},\widetilde{\mathcal{V}}_{S})$ in Eq.~\eqref{eq: wc fid}.
    }
    \label{fig:1}
\end{figure}

In this framework the possible transformations induced on $S$ are given by the following family of CPTP maps \cite{OpenQ.SystemsPetruccione}

\begin{align}  \notag
& {\tt TEP}(H_{B},\beta_B): =  \Big \{ \tcV_S ({\rho}_{S})=\tr_B  \cU_{SB}[{\rho}_{S}\otimes \beta_B],\\
&\qquad\qquad \qquad\qquad\text{s.t.}\; [{U}_{SB},{H}_{SB}] =0 \Big\},
    \label{Def:MainChannel}
\end{align}
depending on the initial state of the battery ${\beta}_B$ and its Hamiltonian $H_{B}$. For ease of notation we also introduced $\mathcal{U}_{SB}[\bullet] := {U}_{SB}\bullet {U}^{\dag}_{SB}$. Since our desire is to implement a fixed target map  $\mathcal{V}_{S}$, we look for choices of ${\beta_B},H_B$  granting the best possible approximation thereof $\tcV_S  \approx \mathcal{V}_{S}$. This requirement can be put on formal grounds by minimizing a divergence between the maps, here the worst-case infidelity
\begin{equation}\label{eq: wc fid}
\epsilon_{wc}(\mathcal{X}_S ,\mathcal{Y}_S) = 1- \min_{\rho_{SA}} F(\mathcal{X}_S[\rho_{SA}], \mathcal{Y}_S[\rho_{SA}])
\end{equation}
where $F(\rho,\sigma)= \left(\tr | \sqrt{\rho} \sqrt{\sigma} |\right)^{2}$ and $A$ is any auxiliary system  (some useful properties of $\epsilon_{wc}$ are summarized in App.~\ref{Sec:ChannelDist}).

There are some conditions that the initial state of the battery should necessarily satisfy for high precision to be possible, often stated in terms of no-go theorems linking $\epsilon_{wc}$ with a certain amount of {\it quantum resources} \cite{QResourceTh} initially present in the battery
\footnote{This minimal quantity is sometimes called the \lq\lq cost'' of implementation of the gate, however, we avoid using this terminology, as only a small fraction of the resource is lost after one use of the battery, implying (as underlined, for instance, in \cite{CataliticCoherence,Chiribella}) that the same battery state can be used a large number of times}.
Different notions of resources $R(\beta_B,H_B)$, determined by  
battery Hamiltonian $H_{B}$ and its state ${\beta}_B$, have emerged in this context \cite{QFIBound,QFI-NonAbelian,Chiribella, aberg2006quantifying}. We are specifically interested in its {\it entropic coherence} which we now define.

{\it Entropic coherence as a resource for gate implementation.---}
For a system $B$ with Hamiltonian $H_B$, let $\mathcal{G}_{H_B}$ be the twirling map $\mathcal{G}_{H_B}[\bullet] := \sum_{E} \Pi_E \bullet \Pi_E$, with $\Pi_E$ running through projectors on the eigenspaces of $H_B$. To a state $\rho_B$ of the system we associate the entropic coherence of energy, given by the following quantity
\begin{equation}\label{eq: ECE}
\rec{\rho_B}{H_B}= S\big(\mathcal{G}_{H_B}[\rho_B]\big) -  S\big(\rho_B\big),
\end{equation}
 where $S$ is the von Neumann entropy (all logarithms are taken in base $2$). Here we focus on energy for convenience, but the definition immediately generalizes to the entropic coherence of any additive quantity associated to a Hermitian operator.  For a single system, $\rec{\rho_B}{H_B}$ is identical to the relative entropy of superposition~\cite{aberg2006quantifying}, computed with respect to energy eigenspaces of $H_B$. Furthermore, when $H_B$ is non-degenerate it also coincides with the relative entropy of coherence~\cite{BaumgratzREC}, computed with respect to the eigenbasis of $H_B$.
However, in both cases a crucial difference appears when considering composite systems. In our case, the entropic coherence $C(\rho_{SB},H_{SB})$ of a joint state of two systems $S$ and $B$ is computed by applying Eq.~\eqref{eq: ECE} to the total Hamiltonian $H_{SB}= H_S +H_B$. In contrast, to compute the joint relative entropy of superposition/coherence one applies the same equation but for the composed twirling map $\mathcal{G}_{H_S}\otimes \mathcal{G}_{H_B}$~\cite{aberg2006quantifying,BaumgratzREC}, which leads to a very different  class of free operations as we discuss below.

We show in App.~\ref{app:RECprop} that EC satisfies the following four conditions, that we identify as prerequisites for any resource for the gate implementation task:  

\begin{itemize} 
    \item[(C1)] It is non-increasing under energy-preserving operations: $ [U_B, {H}_B] = 0 \implies \rec{\rho_B}{H_B} = \rec{\cU_B[\rho_B]}{H_B}$;
    \item[(C2)] It is non-increasing under partial trace: $\rec{\rho_{BA}}{H_{BA}} \geq \rec{\rho_B}{H_B}$, where $A$ is an auxiliary system possibly correlated with $B$, and $\rho_B= \tr_{A} \rho_{BA}$;
    \item[(C3)] It is sub-additive on product states $\rec{\rho_B\otimes \rho_A}{H_{BA}}\leq \rec{\rho_B}{H_B} +\rec{\rho_A}{H_A}$;
    \item[(C4)] It is regular with respect to the trace distance $D$, that is $\underset{D(\rho_{B},\sigma_{B}) \rightarrow 0}{\lim}|\rec{\rho_B}{H_B}-\rec{\sigma_B}{H_B}| =0$.
\end{itemize}
The conditions closely resemble those proposed in \cite{Chiribella,UnivLimitUnitary(Takagi)}. The crucial difference is that in (C1) we explicitly define the energy-preserving operations as the set of free operations.  This rules out the relative entropy of coherence as a valid resource, since for composite systems the latter can increase under generic energy-preserving operations and is non-increasing only under a smaller set of unitary operations. Namely, those of the form $ U_{SB}\ket{E_a}_S\ket{E_b}_B
    =
    e^{i\phi_{a,b}}
    \ket{\pi(E_a,E_b)}_{SB}$, where $\pi(E_{a},E_{b})$ is a permutation of the eigenstates. We prove that these unitaries cannot produce good approximations of general gates on $S$, for example in App.~\ref{App:Other-REC} we show that they are incapable of approximating the qubit Hadamard gate with the worst-case infinitely larger than 50\%. Hence, the relative entropy of coherence is not a relevant resource for gate implementation, and the bound discussed in ~\cite{Chiribella,UnivLimitUnitary(Takagi)} is not applicable.

{\it Batteries suitable for the task.---}
We are interested in approximating a NEPG $\mathcal{V}_S$ on a quantum system with Hamiltonian $H_S$ and dimension $d_S$, by a channel $\widetilde{\mathcal{V}}_S \in {\tt TEP}(H_B,\beta_B)$ induced by a total-energy-preserving operation on the system and the battery (see Eq.~\ref{Def:MainChannel}). In order to discuss the convergence of such channels to the ideal gate, we introduce the following definition. 
\begin{defn}
A family of batteries  $\{H_B^{(\epsilon)},\beta_B^{(\epsilon)}\}_{\epsilon\in(0,\eta)}$
is called {\bf suitable} for $\mathcal{V}_S$ if for every $\epsilon \in (0,\eta)$ there exists a channel
\begin{equation}
    \widetilde{\mathcal{V}}_S \in {\tt TEP}(H_B^{(\epsilon)},\beta_B^{(\epsilon)}) \quad \text{such that} \quad \epsilon_{wc} (\widetilde{\mathcal{V}}_S, \mathcal{V}_{S}) \leq  \epsilon.
\end{equation}
\end{defn}
In words, $H_B^{(\epsilon)}$ and $\beta_B^{(\epsilon)}$ define a battery system capable of implementing $\mathcal{V}_{S}$ with $\epsilon$ precision, as quantified by the worst-case infidelity in Eq.~\eqref{eq: wc fid}. This definition is general and, in particular,  makes no assumptions on the size of the battery system and its scaling with $\epsilon$. However, as we will see, the final result can be significantly strengthened by introducing a mild assumption on the number of energy levels on which the battery state is supported. To do so let 
\begin{equation}\label{eq: energy density scaling}
\mathcal{N}(E,H_{B}) := \suma{E' \leq E} \text{Rank}[\Pi_{E'}]
\end{equation}
be the dimension of the subspace with energy bounded by $E$ for the battery system. Consider the following definition of a {\it proportionate} family of batteries, which has the number of accessible levels scaling at most polynomially with $\epsilon^{-1}$: 
\begin{defn} \label{Def:Proportionate} 
A family of batteries $\{H_B^{(\epsilon)},\beta_B^{(\epsilon)}\}_{\epsilon\in(0,\eta)}$
is called {\bf proportionate} if 
\begin{align}
&\mathcal{N}(2 E_{max}(\beta^{(\epsilon)}_B), H_{B}^{(\epsilon)}) \leq \rm{poly}\left(\epsilon^{-1}\right)\label{eq:Proportionate},
\end{align}
where $E_{max}(\beta_B^{(\epsilon)}):= \max \{E\in \mathcal{E}[H_{B}^{(\epsilon)}]: tr[\beta_B^{(\epsilon)}\Pi_{E}]>0\}$ is the maximal energy on which the state $\beta_B^{(\epsilon)}$ is supported, and $\mathcal{E}[H]$ is the set of eigenvalues of $H$.
\end{defn}
In particular, if the Hilbert space of the battery system $\mathcal{H}_{B}^{(\epsilon)}$ has dimension satisfying $\text{dim}[\mathcal{H}_{B}^{(\epsilon)}]\leq \rm{poly}\left(\epsilon^{-1}\right)$, any family of batteries $\{H_B^{(\epsilon)},\beta_B^{(\epsilon)}\}_{\epsilon\in(0,\eta)}$ over it is proportionate.

{\it Entropic coherence is necessary for NEPG implementation.---}
With these definitions we are ready to state our main result, establishing a strong constraint on the entropic coherence of any battery system suitable to approximate the ideal gate.
\begin{result}[EC required to approximate a NEPG] \label{Th:main}
 (i) If $\{H_B^{(\epsilon)},\beta_B^{(\epsilon)}\}_{\epsilon}$ is a suitable family of batteries for the gate $\mathcal{V}_S $, then
\begin{equation} \label{Eq:RECforAbelian}
   \rec{\beta_B^{(\epsilon)}}{H_B^{(\epsilon)}}
   \geq  \frac{r_{2}(V_{S})}{8}\log(\frac{\sigma(V_{S})}{\epsilon})-o(1).
\end{equation}

(ii) If $\{H_B^{(\epsilon)},\beta_B^{(\epsilon)}\}_{\epsilon}$ is also {\bf proportionate}, then a stronger bound holds
\begin{equation} \label{Eq:RECforAbelian2}
    \rec{\beta_B^{(\epsilon)}}{H_B^{(\epsilon)}} \geq  \frac{r_{2}(V_{S})}{4}\log(\frac{\sigma'(V_{S})}{\epsilon \log^{2}(\epsilon^{-1})})-o(1).
\end{equation}
Here, $r_{2}(V_{S})$ and $\sigma(V_{S})$ only depend on the gate and the spectrum of $H_{S}$, while $\sigma'(V_{S})$ has additional dependence upon the leading exponent of the polynomial in Eq.~\eqref{eq:Proportionate}. Further, $r_{2}(V_{S})\in \{1,\dots,d_{S}^{2}-1\}$, and $\sigma(V_{S}),\sigma'(V_{S})\geq 0$ with equality if and only if $[V_{S},H_{S}]=0$.

(iii) For a qubit system, with $p := |\bra{0}V_S\ket{1}|^2$,

the above bound holds with $r_2(V_S)=2$ \footnote{ However, it is not true that for a 2-dimensional system the quantity defined in Def.\ref{app: defn r lambda new} can be equal to 2, we prove independently such bound for the qubit case thanks to Result \ref{res: res prod qubit}},  $\log(\sigma(V_S))= 4 \log(p) +{\it cst}$ and $\log(\sigma'(V_S))= 2 \log(p)+ {\it cst}$.
\end{result}

Here, $o(1)\to 0$ in the limit $\epsilon\to 0$,  $r_{2}(V_{S})$, defined in  Def.~\ref{app: defn r lambda new} in App.~\ref{App:BoundingProduction}, is connected to the incommensurability rank of the system's energy spectrum $\mathcal{E}[H_{S}]$ \cite{OurIIDEntropy}, 
and $\sigma(V_{S}),\sigma'(V_{S})>0$, defined in Eqs.~(\ref{def:sigma(V_{S})},\ref{def:sigma'(V_{S})}), quantify the amount of asymmetry of $V_{S}$ with respect to $H_{S}$.

The result relies on the facts that $S$ has finite dimension $d_{S}$ and $\mathcal{V}_S$ is unitary. Furthermore, the bounds depend on the gate $\mathcal{V}_S$ and the spectrum of the system Hamiltonian $H_S$. In App.~\ref{App:BoundingProduction} we prove that if $H_{S},V_{S}$ are sampled randomly, then $r_{2}(V_{S})\geq d_{S}-1$ with probability 1. 

{\it Proof sketch of the main result.---}
Inspired by ref. \cite{Chiribella}, we consider $2m$ copies $S_iA_i$ of the system $SA$ with total Hamiltonian $H_{\bm S \bm A}=\underset{i=1}{\overset{2m}{\sum}} H_{S_i}+ H_{A_i}$. The introduced auxiliary systems $A_{i}$ are copies of $S$ and never interact with the battery, their role will be clarified later. By hypothesis, there exists a {\tt TEP} channel $\tilde{\mathcal{V}}_{S}(\cdot):=\tr_B[\mathcal{U}_{SB}(\cdot \otimes \beta_{B})]$ such that $\epsilon_{wc}(\mathcal{V}_{S},\tilde{\mathcal{V}}_{S})\leq \epsilon$. Let the $S_{i}$ systems sequentially interact with a single battery $B$ in state $\beta_{B}$ via $~\mathcal{U}_{S_i B}$ (for odd $i$) and the adjoint unitary $\mathcal{U}^{*}_{S_i B}$ (for even $i$). Since the battery--system interaction is energy-preserving, and the EC satisfies (C1) and (C3), the extra entropic coherence in $\bm{SA}$ must be provided by the battery; formally, for any initial state $\rho_{\bm S \bm A}=\rho_{SA}^{\otimes 2m}$, we have
\begin{align}
C(\beta_B,H_{B}) &\geq  
C\big(\tnu_{\bm S\bm A}, H_{\bm S \bm A}\big) -C \big(\nu_{\bm S\bm A}, H_{\bm S \bm A}\big) \nonumber\\
&+ C\big(\nu_{\bm S\bm A}, H_{\bm S \bm A}\big) 
 - C\big(\rho_{\bm S\bm A}, H_{\bm S\bm A} \big),  \label{eq: main sketch 2}
\end{align}
where $\nu_{\bm S\bm A} = \big((\mathcal{V}_{S}\otimes \mathcal{V}^{*}_{S})^{\otimes m} \otimes \text{id}_{\bm A}\big)[\rho_{\bm S\bm A}]$ is the ``ideal'' final state, and $\tnu_{\bm S\bm A}$ (resulting from the action of $\mathcal{U}_{S_i B}$ and $\mathcal{U}^{*}_{S_i B}$) is guaranteed to be close to the latter $D(\tnu_{\bm S\bm A},\nu_{\bm S\bm A})\leq 4m\sqrt{\epsilon}$~\cite{Chiribella}.

Combining this inequality for $D(\tnu_{\bm S\bm A},\nu_{\bm S\bm A})$ with a refinement of property (C4), based on entropy continuity~\cite{EntropyContinuity,EntropyContinuity2024}, we bound the first line in Eq.~\eqref{eq: main sketch 2}
\begin{align} 
    C \big(&\nu_{\bm S\bm A}, H_{\bm S \bm A}\big) - C\big(\tnu_{\bm S\bm A}, H_{\bm S \bm A}\big) \leq \nonumber \\
    &
    4m\sqrt{\epsilon}
    \left[
        \frac{d_S^2}{2}\log(2m)+4m\log(d_S)
    \right]
+2\,h_2(4m\sqrt{\epsilon}),
    \label{Eq:RegMaintext}
\end{align}
where $h_2$ is the binary entropy. For \textit{proportionate} batteries this inequality can be strengthened to a linear function of $m$, which explains the difference between the two claims of the main result.

To bound the second line in Eq.~\eqref{eq: main sketch 2}, we choose $\rho_{\bm S \bm A}$ to be pure and an eigenstate of $H_{\bm S \bm A}$, such that $C\big(\rho_{\bm S\bm A}, H_{\bm S\bm A} \big)=0$. Then, we identify the quantity $C\big(\nu_{\bm S\bm A}, H_{\bm S \bm A}\big)$ with the entropy of a sum of $2m$ i.i.d. discrete random variables. A lower bound on this entropy has recently been derived in ~\cite{OurIIDEntropy}, and implies the following bound
\begin{align}
C\big(\nu_{\bm S\bm A}, &H_{\bm S \bm A}\big)
 - C\big(\rho_{\bm S\bm A}, H_{\bm S\bm A} \big)  \geq \nonumber \\  
& \frac{r_{2}(V_{S})}{2}\log\big(2 \pi e \; m \; g(V_{S})\big)-o(1),
\label{eq: proof sketch prod}
\end{align}
where $g(V_{S})$ quantifies the asymmetry of the gate and enters in both $\sigma(V_S)$ and $\sigma'(V_S)$ later.

For a qubit system we proceed differently and bound Eq.~\eqref{eq: main sketch 2} with an explicit construction of the initial state $\rho_{\bm S \bm A}$ which is not a product of $2m$ identical states. This yields a tighter but less general bound.

The last step is to combine inequality~\eqref{eq: proof sketch prod}
with \eqref{Eq:RegMaintext} (or a stronger version for \textit{proportionate} families and/or qubit gates) to bound the entropic coherence of the battery $C(\beta_B,H_{B}) $ via Eq.~\eqref{eq: main sketch 2}. Optimizing over the number of copies $m$ gives the final bounds $\square.$

Now, we discuss some implications of the result.\\
{\it Vanishing error requires unbounded battery. --}
The first corollary, given below, is that any finite-dimensional battery is bound to a certain minimal implementation error. To the best of our knowledge, this is the first demonstration of this important physical insight. Additionally, it proves that for a qubit the NEPG implementation scheme with $O(\epsilon^{-\frac{1}{2}})$ battery levels, presented in \cite{UDpaper}, is tight in the exponent.

\begin{cor}[Vanishing error requires unbounded battery]
\label{cor: cor}
     Let $\{H_B^{(\epsilon)},\beta_B^{(\epsilon)}\}_{\epsilon}$ be a suitable family of batteries for the NEPG $\mathcal{V}_S $, then 
    \begin{equation}
       {\rm dim}[H_{B}^{(\epsilon)}] \geq \left(\frac{1-o(1)}{\log(\epsilon^{-1})}\right)^{\frac{r_{2}(V_{S})}{2}} \left(\frac{\sigma'(V_{S})}{\epsilon}\right)^{\frac{r_{2}(V_{S})}{4}} \to \infty
    \end{equation}
    in the limit $\epsilon \to 0$.
\end{cor}

\begin{proof}
If $\{H_B^{(\epsilon)},\beta_B^{(\epsilon)}\}_{\epsilon}$ is not \textit{proportionate}, ${\rm dim}[H_{B}^{(\epsilon)}]$ must grow super-polynomially with $\epsilon^{-1}$ and the bound is true. Otherwise, we obtain the corollary by combining Eq.~\eqref{Eq:RECforAbelian2} with the following immediate upper bound on the coherent entropy of energy
$ \rec{\beta_B^{(\epsilon)}}{H_B^{(\epsilon)}}\leq \log({\rm dim}[H_{B}^{(\epsilon)}])$
implied by the dimension of the state.
\end{proof}

{\it Implications for other resources.---}
Without any assumption on the Hamiltonian of the system the entropic coherence gives no information on the content of other quantities studied for our task, like energy or quantum Fisher information. However, as we now demonstrate, once the spectrum of $H_B$ is fixed, the presence of entropic coherence in the battery implies strong bounds on the presence of other resources. 

In particular, we show in App.~\ref{app: corrs} that for any system with $C(\beta_B,H_B)\geq C$, the average energy satisfies
\begin{align}
    \tr(H_B \beta_B) &\geq \frac{\tr(H_B e^{-\gamma_{C} H_B})}{\tr(e^{-\gamma_{C} H_B})} \label{Eq:minenergy}
\end{align}
where the Lagrangian multiplier $\gamma_{C}$ is determined by the equation 
$  S\left(\frac{ e^{-\gamma_{C} H_B}}{\tr e^{-\gamma_{C} H_B}} \right)=C$~\footnote{Since $S\left(\frac{ e^{-\gamma H}}{\tr e^{-\gamma H}} \right)$ is strictly decreasing in $\gamma$, the solution is always unique.}.

Once  the spectrum of $H_B$ is fixed Eq.~(\ref{Eq:minenergy}) can be easily computed. 
For illustration, we now consider the example of a harmonic oscillator. In this case Eq.~\eqref{Eq:minenergy} implies $\tr [ \beta_B H_{B}]\geq \omega \, \left(\frac{2^{C(\beta,H_{B})}}{e}-\frac{1}{2}+O\left(2^{-C(\beta,H_B)}\right)\right)$ in the large $C(\beta_B,H_{B})$ limit. 
In fact, for any system whose spectral volume grows linearly with the energy
\begin{equation}\label{eq: linear in E}
    \mathcal{N}(H_B^{(\epsilon)},E)\leq 1+\eta E, 
\end{equation}
this simple argument implies the following corollary of the main result (see  App.~\ref{app: corrs} for derivation).

\begin{cor}[Energy constraint to implement NEPGs]\label{Cor:Energy}
Let $\{H_B^{(\epsilon)},\beta_B^{(\epsilon)}\}_{\epsilon}$ be suitable for the NEPG $\mathcal{V}_S$, and $\mathcal{N}(H_B^{(\epsilon)},E)\leq 1+\eta E$ then 
\begin{equation} \label{eq: en lb}
  \tr[\beta_B^{(\epsilon)}H_B^{(\epsilon)}] 
   \geq  \left[\frac{\sigma(V_{S},H_{S})^{\frac{r_{2}(V_{S},H_{S})}{8}}}{e \eta}+o(1)\right]\epsilon^{-\frac{r_{2}(V_{S},H_{S})}{8}} .
\end{equation}
\end{cor}

This corollary is particularly remarkable for systems and gates that operate on non-harmonic energy levels, i.e. such that $r_{2}(V_{S},H_S)>2$ (or $r_{2}(V_{S},H_S)>4$ for a non-proportionate battery). In this case the bound~\eqref{eq: en lb} is stronger in the scaling than all previously known results~\cite{Chiribella, UDpaper}.

We leave open the question whether the extra energy cost required by 
corollary \ref{Cor:Energy} can be avoided with  a battery not satisfying the spectral volume constraint~\eqref{eq: linear in E}. Nevertheless, we speculate
that for $r_{2}(V_{S},H_{S})\geq 2$, a more energy efficient battery system would be composed of multiple harmonic oscillators with frequencies resonant with each transition.

A similar argument holds for the variance and its convex roof, i.e. the Quantum Fisher Information (QFI)~\cite{QFI-ConvexRoof}. In App. \ref{app: corrs} we prove the following
\begin{cor}[QFI constraints to implement NEPGs]\label{Cor:QFI}
Let $H_{B}$ be a harmonic oscillator of frequency $\omega$ and $\{H_B,\beta_B^{(\epsilon)}\}_{\epsilon}$ be suitable for the NEPG $\mathcal{V}_S $, then 
\begin{align} \label{eq: QFI bound corr}
  {\rm QFI}(\beta,H_{B})\geq &\left(\frac{2\omega^{2}\sigma(V_{S},H_{S})^{\frac{r_{2}(V_{S},H_{S})}{4}}}{e\pi}-o(1)\right)\nonumber \\
  &\times (\epsilon d_{S})^{-\frac{r_{2}(V_{S},H_{S})}{4}}.
\end{align}
\end{cor} 
We believe that this result can be generalized up to a constant to all batteries whose energies do not concentrate $\mathcal{N}(E+ \omega, H_{B})-\mathcal{N}(E, H_{B})\leq \eta  \omega$. In \cite{QFIBound}, it  was shown that there always exist batteries $\{\beta^{\epsilon},H_{B}^{\epsilon}\}$ suitable for $V_{S}$ with $ {\rm QFI}(\beta,H_{B})\simeq \epsilon^{-1}$. Thus, the corollary demonstrates that for systems with $r_{2}(V_{S},H_{S})>2$ (or $r_{2}(V_{S},H_{S})>4 $ in the non-proportionate case) harmonic oscillators make very sub-optimal batteries.

{\it Conclusions.---}  

 We investigated the role of entropic coherence of energy as a fundamental resource for the implementation of non-energy-preserving gates (NEPGs). 
We established that, to implement any NEPG with accuracy $\epsilon$, a minimum amount of entropic coherence scaling as $r\log(\epsilon^{-1})$ is required in the battery, where $r \geq 1/8$ depends on the specific system and gate. Under mild assumptions on the battery, this bound is strengthened by a factor of two and is tight in the case of qubits.

We then examine the implications of this bound. First, we show that the dimension of the battery system must scale at least as $d_B \sim \epsilon^{-2r}$, thus diverging to enable perfect implementation. Furthermore, we demonstrated that the battery's energy and Quantum Fisher Information (QFI) must scale similarly, provided the Hamiltonian has a bounded density of energy levels. These bounds can be stricter than those derived independently in \cite{Chiribella,QFIBound}, suggesting a hierarchy among physical resources. Notably, constraints on energy and QFI do not imply corresponding limits on EC, since, by rescaling the battery's Hamiltonian, both energy and QFI can become arbitrarily large while EC remains unchanged.
Finally, our results offer criteria for evaluating the efficiency of batteries in implementing NEPGs. In particular, for systems with multiple incommensurable energy gaps, batteries with a bounded energy level density are shown to be highly inefficient. We speculate that a battery composed of multiple harmonic oscillators—one for each energy gap—could instead prove to be efficient.

Future developments could aim to tighten the inequality by a factor of two for arbitrary systems, by extending the strategy used in the qubit case, namely selecting an appropriately entangled initial state during the resource production stage of the proof. Another promising direction is to derive upper bounds on the EC cost by explicitly constructing battery states and interactions that implement non-energy-preserving gates beyond the already explored case of equally spaced energy levels \cite{CataliticCoherence,Chiribella,UDpaper}.

{\it Acknowledgments :--}
RC and PS acknowledge support from the Swiss National Science Foundations (project 219366 and NCCR-SwissMAP).
VC and VG acknowledge financial support by MUR (Ministero dell’ Università e della Ricerca) through the PNRR MUR project PE0000023-NQSTI. M.P.-L. acknowledges support from the Grant RYC2022-036958-I funded by MICIU/AEI/10.13039/501100011033 and by ESF+.

\bibliography{biblio}

\clearpage

\appendix

\begin{widetext}

{\color{black}
\tableofcontents
}

\vspace{10 mm}

\section{Channel distances and their relationships} \label{Sec:ChannelDist}

To evaluate how successful the setup in Fig. \ref{fig:1} was we need a way to measure how close the maps $\widetilde{\mathcal V}_S \in {\tt TEP}(H_{B},\beta_B)$ and $\mathcal{V}_S$ are to each other. 
This can be done by using several {channel distances} between two quantum channels $\cE_1$ and $\cE_2$ acting on the system $X$. A commonly used distance is the {\it diamond norm}
\begin{equation} \label{Def:DiamondNorm}
    \lVert\cE_1- \cE_2 \rVert_{\diamond}
    :=\underset{\rho_{XA}\in\mathcal{D}(\mathcal{H}_{XA})}{\operatorname{sup}}
    \left\lVert
    \big((\cE_1-\cE_2)\otimes\operatorname{id}_A\big)[\rho_{XA}]
    \right\rVert_1
    =2\underset{\rho_{XA}\in\mathcal{D}(\mathcal{H}_{XA})}{\operatorname{sup}}
    D\Big((\cE_1\otimes\operatorname{id}_A)[\rho_{XA}],
    (\cE_2\otimes\operatorname{id}_A)[\rho_{XA}]\Big).
\end{equation}
where the channels are trivially extended to act on an auxiliary system $A$ (that can be taken of the same dimension as $X$ without loss of generality), the supremum is taken over all states $\rho_{XA}$ of the composed systems $XA$, and $D(\rho,\sigma) = \frac{1}{2} \tr |\rho-\sigma|$ is the trace distance. For two equiprobable channels, the optimal single-shot discrimination probability is $\frac{1}{2}+\frac{1}{4}\lVert\cE_1-\cE_2\rVert_\diamond$.

Similarly, the notion of fidelity between quantum states can be generalized to channels by introducing the {\it worst-case fidelity}
\begin{equation} \label{eq:WCinfdef}
    F_{wc}(\cE_1,\cE_2):=\underset{\rho_{SA}\in\mathcal{D}(\mathcal{H}_{SA}) }{\text{inf}} F\Big( (\cE_1\otimes \text{id}_A)[{\rho}_{SA}],(\cE_2\otimes \text{id}_A)[{\rho}_{SA}] \Big),
\end{equation} 
where we use the definition of fidelity with the square $F(\rho,\sigma) = \left(\tr |\sqrt{\rho}\sqrt{\sigma}|\right)^2$.
Using $F_{wc}$ we can define the {\it worst-case infidelity}
\begin{equation} \label{Def:WCInfidelity}
    \epsilon_{wc}(\cE_1,\cE_2):=1- F_{wc}(\cE_1,\cE_2).
\end{equation}
Notice that $\epsilon_{wc}(\cE_1,\cE_2)$ is not a distance because it fails to satisfy the triangle inequality,
however one can verify that the worst-case angle $A_{wc}(\cE_1,\cE_2):=\arccos(\sqrt{F_{wc}(\cE_1,\cE_2)})$ is instead a distance.
The worst-case infidelity is the error quantifier used in our definition of the gate approximation task. It is related to the diamond norm by the following inequalities
\begin{equation}
 2\left(1- \sqrt{1-\epsilon_{wc}(\cE_1,\cE_2)}\right)
 \leq \lVert  \cE_1-\cE_2 \rVert_{\diamond}
 \leq 2\sqrt{\epsilon_{wc}(\cE_1,\cE_2)},
\end{equation}
which is a direct consequence of the Fuchs--van de Graaf inequality  $ 1-\sqrt{F({\rho},{\tau})}\leq D({\rho},{\tau})\leq \sqrt{1-F({\rho},{\tau})}$ \cite{FuchsInequality}.

\section{The entropic coherence is a resource for the gate implementation}
\label{app:RECprop}
We start by introducing the set of incoherent states.
Let $A,B$ be two systems
with respective Hamiltonians $H_{A},H_{B}$; the incoherent states are defined as 
\begin{align} \label{eq:IncoherentSet}
    \Delta_{AB} & :=\{{\rho}_{AB}: C({\rho}_{AB},{H}_{A}+{H}_{B})=0\};\\
    \Delta_{B}& :=\{{\rho}_{B}: C({\rho}_{B},{H}_{B})=0\}. \notag
\end{align}
It is immediate to verify that every incoherent state of a given Hamiltonian $H$ is a fixed point of the corresponding twirling channel $\mathcal{G}_{H}$ i.e., if $\rho \in \Delta$ then  $\mathcal{G}_{H}(\rho)=\rho$, and that $\mathcal{G}_{H}(\sigma) \in \Delta \quad \forall \sigma \in \mathcal{D}[\mathcal{H}] $, making it a projector into the incoherent states set.
It is also evident that $\rho$ is an incoherent state if and only if it does not evolve under the action of $e^{-i Ht}$, i.e. $\rho\in \Delta \Longleftrightarrow \rho= e^{-iHt} \rho e^{iHt}$ for all $t$.
Further, jointly incoherent states are locally incoherent:
\begin{lem}\label{lem:incoherent}
    Let $\Delta_{AB}$, $\Delta_{B}$ be the sets of incoherent states as \eqref{eq:IncoherentSet}, then 
    $Tr_{A}[\Delta_{BA}]=\Delta_{B}$.
\end{lem}
\begin{proof}
The inclusion $Tr_{A}[\Delta_{AB}] \supseteq \Delta_{B}$ is trivial. On the other hand given $\rho_{AB}\in \Delta_{AB}$ we have that for all $t$ it holds $\rho_{AB}=e^{-it(H_{A}+H_{B})}\rho_{AB}e^{+it(H_{A}+H_{B})}$. Taking the partial trace over $A$ on both sides we have $\rho_{B}=Tr_{A}[e^{-it(H_{A}+H_{B})}\rho_{AB}e^{+it(H_{A}+H_{B})}]=e^{-itH_{B}}\rho_{B}e^{+itH_{B}}$, proving that $\Delta_{B} \supseteq Tr_{A}[\Delta_{AB}]$ and so the thesis.
\end{proof}

We now prove that the entropic coherence admits three equivalent definitions:
\begin{align}
    & (i) \quad C(\rho,H):= S(\mathcal{G}_{H}(\rho))- S(\rho), \\
    & (ii) \quad C(\rho,H)=S(\rho||\mathcal{G}_{H}(\rho)),   \\ 
    & (iii)\quad C(\rho,H)= \underset{\sigma \in \Delta}{\min} S(\rho||\sigma).
\end{align}
\begin{proof}
 We use the identity $\log(\sigma)=\log(\mathcal{G}_{H}(\sigma))=\mathcal{G}_{H}(\log(\sigma))$, and the fact that $\mathcal{G}_{H}$ is self-adjoint, to prove 
\begin{align}
   & S(\rho || \sigma)=-Tr[\rho\log(\sigma)] -S(\rho)=-Tr[\rho \mathcal{G}_{H}(\log(\sigma))]-S(\rho) \nonumber  \\
   & = -Tr[\mathcal{G}_{H}(\rho) \log(\sigma)]-S(\rho)=S(\mathcal{G}_{H}(\rho))-S(\rho)+S(\mathcal{G}_{H}(\rho)||\sigma).
\end{align}
This equality proves $(i)\iff (ii)$ since for $\sigma=\mathcal{G}_{H}(\rho)$ one has $S(\mathcal{G}_{H}(\rho)||\sigma)= 0$,  and $(i)\iff (iii)$  because 
for all other choices of $\sigma$ one has $S(\mathcal{G}_{H}(\rho)||\sigma)> 0$.  
\end{proof}
Notice that in particular the expression $(iii)$ makes it explicit that the Entropic Coherence can be seen as the relative entropy of resource, where $\Delta$ is the convex set of free states. This guarantees that (EC) satisfies many desired properties as a coherence quantifier \cite{Corrected-Uni-Conv-Rate(2012), PlenioCoherence}. We now prove that it additionally satisfies $C_{1}-C_{4}$ of the main text: 

\begin{proof}
    First, since $\mathcal{G}_{H}$ is a unital channel, $C(\rho,H)\geq 0$ \cite{UnitalEntro}.\\
(C1) If $[U,H]=0$ then $ \mathcal{G}_{H} \circ \mathcal{U}= \mathcal{U}\circ\mathcal{G}_{H}$ thus
\begin{align}
     &C(\mathcal{U}(\rho),H)=S(\mathcal{G}_{H} \circ \mathcal{U}(\rho))-S(\mathcal{U}(\rho))\\
     &=S(\mathcal{G}_{H} (\rho))-S(\rho)=C(\rho,H).
\end{align}

(C2) From lemma \ref{lem:incoherent} we know that $\Delta_{A} = tr_{B}[\Delta_{AB}]$, thus we have
\begin{align}
   & C(\rho_{AB},H_{AB})=\underset{\sigma \in \Delta_{AB}}{\min} S(\rho_{AB}||\sigma)\geq 
   \underset{\sigma \in \Delta_{AB}}{\min} S(\rho_{A}||tr_{B}[\sigma]) 
   \\
    &= \underset{\sigma \in \Delta_{A}}{\min} S(\rho_{A}||\sigma):=C(\rho_{A},H_{A})
\end{align}
 where in the first inequality we used the data processing inequality of the relative entropy \cite{RelativeEntroIneq} $S(\mathbf{\Phi}(\rho)||\mathbf{\Phi}(\sigma))\leq S(\rho||\sigma)$.

(C3) First observe that $\Delta_{A}\otimes \Delta_{B}\subseteq \Delta_{AB}$, i.e. the tensor product of incoherent states is still incoherent with respect to the joint system; this implies, according to characterization (iii),
\begin{align}
    &S(\rho_{A}\otimes \sigma_{B} ||\mathcal{G}_{H_{AB}}(\rho_{A}\otimes \sigma_{B}))\leq S(\rho_{A}\otimes \sigma_{B} ||\mathcal{G}_{H_{A}}(\rho_{A})\otimes \mathcal{G}_{H_{B}}(\sigma_{B})) \\
    &= S(\rho_{A}||\mathcal{G}_{H_{A}}(\rho_{A}))+ S(\sigma_{B} ||\mathcal{G}_{H_{B}}(\sigma_{B})).
\end{align}

(C4) Because the twirled states are closer in trace distance than their counterparts, we can apply the Fannes-Audenaert inequality \cite{EntropyReg} twice:
\begin{align} \label{Eq:RECRegularity}
    &\abs{C(\rho,H)-C(\sigma,H)}\leq \abs{S(\overline{\rho})-S(\overline{\sigma})} +\abs{S(\rho)-S(\sigma)} \\
    & \leq 2\log(d-1)D(\rho,\sigma)+2h_{2}
     (D(\rho,\sigma)),
\end{align} 
where $d:=\dim(\mathcal{H})$ and $h_{2}(x):=-x\log(x)-(1-x)\log(1-x)$ is the binary entropy. 
\end{proof}

\section{Comment on the bound on the relative entropy of coherence in refs \cite{Chiribella,UnivLimitUnitary(Takagi)}} \label{App:Other-REC}
In \cite{Chiribella,UnivLimitUnitary(Takagi)}, a theorem analogous to our result was derived for a similarly looking quantity, the relative entropy of coherence (REC). However, the definition of REC employed in these works differs from ours in a subtle way that makes it considerably less relevant in the context of NEPG implementation, as we now explain.

Let $S$ and $B$ be two quantum systems with respective Hamiltonians $H_S$ and $H_B$, and let $\mathcal{G}_{H_S}$ and $\mathcal{G}_{H_B}$ denote the complete dephasing channels in fixed eigenbases of the two Hamiltonians. According to the definition used in Refs.~\cite{Chiribella,UnivLimitUnitary(Takagi)}, the REC of a composite system, which we denote by $C_{\rm loc}(\cdot)$, reads
\begin{equation}
    C_{\rm loc}(\rho_{SB})
    :=
    S\!\left[
    \left(\mathcal{G}_{H_S}\otimes\mathcal{G}_{H_B}\right)(\rho_{SB})
    \right]
    -S(\rho_{SB}).
\end{equation}

Notice that, according to this definition, the knowledge of the total Hamiltonian $H_{SB}$ and of the joint state $\rho_{SB}$ does not uniquely specify the value of $C_{\rm loc}(\rho_{SB})$: the decomposition of the global system into subsystems must also be specified. Furthermore, in Refs.~\cite{Chiribella,UnivLimitUnitary(Takagi)}, our condition C1 is replaced by the requirement that a resource be non-increasing under free operations and partial trace, while free operations are defined as those that cannot increase the resource.

Let $\{\ket{E_a}_S\}_a$ and $\{\ket{E_b}_B\}_b$ denote the fixed dephasing bases. A unitary is free if and only if
\begin{align}
    &C_{\rm loc}\!\left(U_{SB}\rho_{SB}U_{SB}^{\dagger}\right)
    \leq C_{\rm loc}(\rho_{SB})
    \qquad \forall\,\rho_{SB}
    \nonumber\\
    &\Longleftrightarrow\quad
    U_{SB}\ket{E_a}_S\ket{E_b}_B
    =
    e^{i\phi_{a,b}}
    \ket{\pi(E_a,E_b)}_{SB},
    \label{eq:free-monomial-unitaries}
\end{align}
where $\pi$ is an arbitrary permutation of the joint product basis and $\{\phi_{a,b}\}_{a,b}$ are arbitrary real phases. Indeed, every product-basis vector has vanishing $C_{\rm loc}$ and must therefore be mapped to another pure state with vanishing $C_{\rm loc}$, namely another product-basis vector. Unitarity then implies that this correspondence is a permutation. Conversely, every unitary of the form~\eqref{eq:free-monomial-unitaries} satisfies
\begin{equation}
    \left(\mathcal{G}_{H_S}\otimes\mathcal{G}_{H_B}\right)
    \!\left(U_{SB}\rho_{SB}U_{SB}^{\dagger}\right)
    =
    U_{SB}
    \left(\mathcal{G}_{H_S}\otimes\mathcal{G}_{H_B}\right)
    (\rho_{SB})
    U_{SB}^{\dagger},
\end{equation}
and hence preserves $C_{\rm loc}$ for every state.

Let us denote by ${\tt LEP}(\beta_B,H_B)$ the set of channels induced on $S$ by such free unitaries, analogously to Eq.~\eqref{Def:MainChannel}. This set of operations is considerably more constrained than ${\tt TEP}(\beta_B,H_B)$, defined through the total-energy-preserving condition $[U_{SB},H_S+H_B]=0$. Although the free unitaries in Eq.~\eqref{eq:free-monomial-unitaries} may exchange energy between $S$ and $B$, the corresponding set ${\tt LEP}(\beta_B,H_B)$ does not contain arbitrarily accurate approximations of all unitary gates on $S$, regardless of the choice of the auxiliary system $B$. More precisely, we have the following.

\begin{lem}\label{lem:Hadamard-LEP}
Let $S$ be a qubit with non-degenerate Hamiltonian $H_S$, for every auxiliary system $B$, every state $\beta_B$, and every channel
$\widetilde{\mathcal{V}}_S\in{\tt LEP}(\beta_B,H_B)$, it holds that
\begin{equation}
    \epsilon_{wc}
    \left(
    \widetilde{\mathcal{V}}_S,
    \mathcal{V}_{\rm Had}
    \right)
    \geq \frac{1}{2}.
    \label{eq:Hadamard-LEP-bound}
\end{equation}
where $\mathcal{V}_{\rm Had}$ is the channel corresponding to the Hadamard gate 
\begin{equation}
    V_{\rm Had}
    :=
    \frac{1}{\sqrt{2}}
    \begin{pmatrix}
        1 & 1\\
        1 & -1
    \end{pmatrix}
\end{equation}

\end{lem}

\begin{proof}
Consider a spectral decomposition of the battery state,
\begin{equation}
    \beta_B
    =
    \sum_{\ell}q_{\ell}
    \ket{\beta_{\ell}}\!\bra{\beta_{\ell}},
\end{equation}
and let $\{\ket{E_c}_B\}_c$ be the fixed dephasing basis of $B$. The channel induced on $S$ admits the Kraus decomposition
\begin{align}
    \widetilde{\mathcal{V}}_S(\rho_S)
    &=
    \sum_{c,\ell}
    K_{c\ell}\rho_S K_{c\ell}^{\dagger},
    \label{eq:LEP-Kraus}\\
    K_{c\ell}
    &:=
    \sqrt{q_{\ell}}\,
    {}_B\!\bra{E_c}
    U_{SB}
    \ket{\beta_{\ell}}_B,
    \qquad
    \sum_{c,\ell}K_{c\ell}^{\dagger}K_{c\ell}
    =
    \mathbbm{1}_S.
    \label{eq:LEP-Kraus-definition}
\end{align}

The matrix elements of these Kraus operators in the eigenbasis
$\{\ket{0},\ket{1}\}$ of $H_S$ are
\begin{equation}
    (K_{c\ell})_{ts}
    =
    \bra{t}K_{c\ell}\ket{s}
    =
    \sqrt{q_{\ell}}
    \sum_b
    \bra{t,E_c}U_{SB}\ket{s,E_b}
    \braket{E_b |\beta_{\ell}}.
    \label{eq:LEP-Kraus-elements}
\end{equation}
 Writing the permutation
appearing in Eq.~\eqref{eq:free-monomial-unitaries} as
$\pi(s,b)=(\pi_S(s,b),\pi_B(s,b))$. Then
\begin{equation}
    U_{SB}\ket{s,E_b}
    =
    e^{i\phi_{s,b}}
    \ket{\pi_S(s,b),E_{\pi_B(s,b)}},
\end{equation}
and therefore
\begin{equation}
    \bra{t,E_c}U_{SB}\ket{s,E_b}
    =
    e^{i\phi_{s,b}}
    \delta_{t,\pi_S(s,b)}
    \delta_{c,\pi_B(s,b)}.
    \label{eq:free-unitary-elements}
\end{equation}
Since $\pi$ is a permutation, for every fixed output pair $(t,c)$
there is a unique input pair $(s_{t,c},b_{t,c}):=\pi^{-1}(t,c)$ for which the matrix element is not zero. 

Substituting Eq.~\eqref{eq:free-unitary-elements} into
Eq.~\eqref{eq:LEP-Kraus-elements} gives
\begin{equation}
    (K_{c\ell})_{ts}
    =
    \sqrt{q_{\ell}}\,
    e^{i\phi_{s_{t,c},b_{t,c}}}
    \braket{E_{b_{t,c}}|\beta_{\ell}}\,
    \delta_{s,s_{t,c}}.
    \label{eq:LEP-Kraus-explicit}
\end{equation}
Thus, for fixed $c$ and $\ell$, the row of $K_{c\ell}$ labelled by
$t$ can have a nonzero entry only in the column $s_{t,c}$. Hence each
row contains at most one nonzero entry and $K_{c\ell}$ contains at
most two nonzero entries in total. This observation already tells us that the Krauss operators cannot be proportional (or close to proportional) to the target $V_{\rm Had}$, thus the corresponding channel is not a good approximation of the Hadamard gate. To prove this quantitatively, first notice that the Choi fidelity can be computed in terms of Krauss operators as 

\begin{align} \label{Eq:Fchoi}
    F_{\rm Choi}
    &:=
    F\left(
    \left(\mathcal{I}_A\otimes\widetilde{\mathcal{V}}_S\right)
    [\ketbra{\Phi}{\Phi}],
    \ketbra{\Phi_{\rm Had}}{\Phi_{\rm Had}}
    \right)=
    \frac{1}{4}
    \sum_{c,\ell}
    \left|
    \tr\!\left[V_{\rm Had}^{\dagger}K_{c\ell}\right]
    \right|^2  
\end{align}
Where we defined 
\begin{equation}
    \ket{\Phi}_{AS}
    :=
    \frac{\ket{00}+\ket{11}}{\sqrt{2}},
    \qquad
    \ket{\Phi_{\rm Had}}_{AS}
    :=
    \left(\mathbbm{1}_A\otimes V_{\rm Had}\right)\ket{\Phi}_{AS}.
\end{equation}

To compute the latter, define  

\begin{equation}
    \mathcal{I}_{c\ell}
    :=
    \left\{
    (t,s):(K_{c\ell})_{ts}\neq0
    \right\},
    \qquad
    \abs{\mathcal{I}_{c\ell}}\leq2,
\end{equation}
and let
\begin{equation}
    h_{00}=h_{01}=h_{10}=1,
    \qquad
    h_{11}=-1.
\end{equation}
By the Cauchy--Schwarz inequality,
\begin{align}
    & \left|
    \tr\!\left[V_{\rm Had}^{\dagger}K_{c\ell}\right]
    \right|^2=\frac{1}{2}\left|
    (K_{c\ell})_{00}
    +(K_{c\ell})_{01}
    +(K_{c\ell})_{10}
    -(K_{c\ell})_{11}
    \right|^2=\frac{1}{2}
    \left|
    \sum_{(t,s)\in\mathcal{I}_{c\ell}}
    h_{ts}(K_{c\ell})_{ts}
    \right|^2
    \nonumber\\
    &\leq \frac{1}{2}
    \abs{\mathcal{I}_{c\ell}}
    \sum_{(t,s)\in\mathcal{I}_{c\ell}}
    \left|(K_{c\ell})_{ts}\right|^2=
    \sum_{t,s=0}^{1}
    \left|(K_{c\ell})_{ts}\right|^2= \tr\!\left[K_{c\ell}^{\dagger}K_{c\ell}\right].
    \label{eq:Hadamard-entry-bound}
\end{align}
Thus using Eq.\eqref{Eq:Fchoi} we have 
\begin{align}
  & F_{\rm Choi}\leq
    \frac{1}{4}
    \sum_{c,\ell}
    \tr\!\left[K_{c\ell}^{\dagger}K_{c\ell}\right]
    =
    \frac{1}{4}\tr[\mathbbm{1}_S]
    =
    \frac{1}{2}.
    \label{eq:Hadamard-Choi-bound}
\end{align}

Since the worst-case fidelity cannot exceed its value on this
particular input,
\begin{align}
    1-
    \epsilon_{wc}
    \left(
    \widetilde{\mathcal{V}}_S,
    \mathcal{V}_{\rm Had}
    \right)
    &\leq F_{\rm Choi}
    \leq \frac{1}{2},
\end{align}
which proves Eq.~\eqref{eq:Hadamard-LEP-bound}.
\end{proof}
It should be noted that a similar conclusion can be derived for any non-diagonal gate with a nonzero diagonal. 

In view of this, it should not come as a surprise that the bound found in Eq.~(10) of Ref.~\cite{Chiribella} is exponentially stronger than ours:
\begin{equation}
    \label{Eq:ChiribellaRECbound}
    C_{\rm loc}(\beta_B,H_B)
    \geq
    \frac{f(V_S)}{\sqrt{\epsilon_{wc}}}
    -o(1).
\end{equation}
Indeed, already qubit gates cannot be approximated arbitrarily well by channels in ${\tt LEP}(\beta_B,H_B)$, irrespective of the choice of $\beta_B$ and $H_B$. On the other hand, Eq.~\eqref{Eq:ChiribellaRECbound} is formulated in the limit $\epsilon_{wc}\rightarrow0$, making the corresponding asymptotic bound inapplicable for the generic gate.

\section{Proof of the main result}
 
This section is devoted to the proof of the main result in the main text.

\label{app:Recproof}
\subsection{Resource inequalities and proof outline} \label{Sec:Resource inequalities}

In this section we will present a  generalization of the method used in \cite{Chiribella,UnivLimitUnitary(Takagi)} and show how to use a resource satisfying the properties (C1)-(C4) to create a bound on the precision in the implementation of NEPGs.
To this end, let us consider a battery $B$ in contact with a system $S$ and an auxiliary system $A$ (this will be convenient later on). Let the joint system $SA$ be prepared in the state $\rho_{SA}$ and the battery in $\beta_B$, as usual. Further, we always suppose the energy to be additive, i.e. $H_{SBA}:= H_{S}+H_{B}+H_{A}$. We now prove the following lemma.

\begin{lem} \label{lem: REC bound gen}
    Consider a system with Hamiltonian $H_S$ and a battery with Hamiltonian $H_B$, prepared in the state $\beta_B$. Let $\mathcal{U}_{SB}[\bullet]=U_{SB}\bullet U_{SB}^\dag$ be a joint unitary channel, commuting with the total Hamiltonian $[U_{SB},H_S+H_B]=0$, and $\tcV_S[\bullet] = \tr_B \mathcal{U}_{SB}[\bullet \otimes \beta_B]$ be the channel induced by this unitary on the system (an element of ${\tt TEP}(H_{B},\beta_B)$). Then, for any initial state $\rho_{SA}$ of the system $S$ and an auxiliary system $A$, with Hamiltonian $H_A$, the following bound holds
 \begin{align} \label{eq:WHATEVER}
    C\big(\beta_B,H_{B}\big) \geq
    C\big(\tcV_S[\rho_{SA}], H_S+H_A\big) - C\big(\rho_{SA}, H_{S}+H_A\big). 
\end{align}
\end{lem}
\begin{proof}

Using the properties (C3), (C1) we can write
\begin{align} \label{eq:1Pcoh}
 C(\beta_B,H_{B}) + C(\rho_{SA},H_{SA})  \geq C(\beta_B \otimes \rho_{SA},H_{SBA})  
     = C\big((\mathcal{U}_{SB}\otimes \text{id}_A)[\beta_B \otimes \rho_{SA}],H_{SBA}\big)
\end{align}

Using the definition of $\tcV_S$ as the trace over $B$ of the global evolution and property (C2) we obtain
\begin{align} \label{eq:2Pcoh}
 &  C\big((\mathcal{U}_{SB}\otimes \text{id}_A)[\beta_B \otimes \rho_{SA}],H_{SBA}\big) \geq C(\tcV_S \otimes \mathcal{I}_A (\rho_{AS}),H_{SA}).
 \end{align}

\end{proof}
The general idea now is to somehow apply this bound to our channel of interest $\tcV_S$   approximating $\mathcal{V}_S$. However, following \cite{Chiribella} we note that a tighter result can be obtained when applying the bound to several copies of this channel. Therefore, we now consider $2m$ copies of the system $S$, labeled $S_i$ for $i \in \{1, ... 2m\}$, and denote their composite system with $\bm S =S_1\dots S_{2m}$. For any (single copy) total-energy-preserving unitary channel $\mathcal{U}_{SB}$, its adjoint $\mathcal{U}^*_{SB}[\bullet]=U_{SB}^\dag \bullet U_{SB}$, and any state of the battery $\beta_B$ we define the following channels 
\begin{align} \label{eq:manyU}
&\mathcal{U}^{(m)}_{\bm S B}:=\mathcal{U}_{S_{2m}B}^*\circ \mathcal{U}_{S_{2m-1}B}\circ...\circ \mathcal{U}_{S_{2}B}^*\circ\mathcal{U}_{S_{1}B}\\
& \tcV_{\bm S}^{(m)} [\bullet_{\bm S}] :=\tr_B \mathcal{U}^{(m)}_{\bm S B}[\bullet_{\bm S} \otimes \beta_B],
\end{align}
where $\mathcal{U}^{(m)}_{\bm S B}$ is a unitary channel on $\bm S B$, 

and $\tcV_{\bm S}^{(m)}$ 

is a CPTP map on $\bm S$. Similarly, let us define the target unitary channel on the $2m$ system copies
\begin{equation} \label{eq:manyV}
   \mathcal{V}^{(m)}_{\bm S}:=\underset{i=1}{\overset{m}{\bigotimes}} \, \mathcal{V}_{S_{2i-1}} \otimes \mathcal{V}_{S_{2i}}^{*}.
\end{equation}

Our choice for these specific definitions of the many-copy maps $\mathcal{U}_{\bm S B}^{(m)}$ and  $\mathcal{V}^{(m)}_{\bm S}$ in Eqs.~(\ref{eq:manyU}-\ref{eq:manyV}) is motivated by the following theorem, relating their distance to the single-copy error $\epsilon_{wc}(\tcV_S;\cV_{S})$.

\begin{thm}
\cite{Chiribella} Consider the maps defined in Eqs. \eqref{eq:manyU} and \eqref{eq:manyV}. For the following  CPTP maps from $L(\cH_{\bm S})\to L(\cH_{\bm S}\otimes \cH_B)$,  where $L(\mathcal{H})$ denotes the space of linear operators acting on the Hilbert space $\mathcal{H}$, we have

 \begin{equation} \label{Th3}
    \lVert \mathcal{U}^{(m)}_{\bm S B}[\bullet_{\bm S} \otimes \beta_B]- \mathcal{V}^{(m)}_{\bm S}[\bullet_{\bm S}] \otimes \beta_B \rVert_{\diamond} \leq 8 m \sqrt{\epsilon_{wc}(\tcV_S;\cV_{S})}.
\end{equation}

\end{thm}
Thanks to this theorem and using the definition of diamond norm in Eq. \eqref{Def:DiamondNorm} we have
\begin{equation} \label{Def:DiamondNormused} D\left(
\mathcal{U}^{(m)}_{\bm S B}[{\rho}^{(m)}_{\bm S \bm A}\otimes \beta_B], \mathcal{V}^{(m)}_{\bm S}[\rho^{(m)}_{\bm S \bm A}]\otimes \beta_B \right) \leq  4 m \sqrt{\epsilon_{wc}(\tcV_S;\cV_{S})},
\end{equation}
that holds for any state $\rho^{(m)}_{\bm S \bm A}$ of the system $\bm S$ extended with any auxiliary system $\bm A$. Since the trace distance is non-increasing under partial trace, this inequality still holds when the battery is traced out
\begin{equation} \label{eq:distance4m}
   D\left(\tcV_{\bm S}^{(m)}[\rho_{\bm S\bm A}^{(m)}],\cV_{\bm S}^{(m)}[\rho_{\bm S \bm A}^{(m)}]  \right) = D\left( \widetilde\nu_{\bm S \bm A}^{(m)},\nu_{\bm S\bm A}^{(m)} \right)   \leq  4 m \sqrt{\epsilon_{wc}(\tcV_S;\cV_{S})}.
\end{equation}
Here we introduced the states 
\begin{equation}
   \widetilde{\nu}_{\bm S\bm A}^{(m)}: =\tcV_{\bm S}^{(m)}[\rho_{\bm S\bm A}^{(m)}] \qquad \text{and}\qquad  \nu_{\bm S\bm A}^{(m)}: = \mathcal{V}_{\bm S}^{(m)}[\rho^{(m)}_{\bm S \bm A}],
\end{equation}
which implicitly depend on $\rho^{(m)}_{\bm S \bm A}$ and are central for the following discussion.

In addition, note that the global unitary $\mathcal{U}_{\bm SB}^{(m)}$ manifestly commutes with the total Hamiltonian of the $2m$ systems and the battery $H_{\bm S} +H_{B}$, with $H_{\bm S} =\sum_{i=1}^{2m} H_{S_i}$. Hence, Lemma~\ref{lem: REC bound gen} immediately implies the following bound
\begin{align}\label{eq: LB-UB}
C(\beta_{B},H_{B}) &\geq C\big(\tcV^{(m)}_{\bm S}[\rho^{(m)}_{\bm S\bm A}], H_{\bm S}+H_{\bm A}\big) - C\big(\rho^{(m)}_{\bm S\bm A}, H_{\bm S}+H_{\bm A} \big)   \\
&= \label{eq: term to LB} C\big(\nu_{\bm S\bm A}^{(m)}, H_{\bm S}+H_{\bm A}\big) 
 - C\big(\rho_{\bm S\bm A}^{(m)}, H_{\bm S}+H_{\bm A} \big) \\
 \label{eq: term to UB}
&-\left(C\big(\nu_{\bm S\bm A}^{(m)}, H_{\bm S}+H_{\bm A}\big) -C\big(\tnu_{\bm S\bm A}^{(m)}, H_{\bm S}+H_{\bm A}\big) \right) 
\end{align}
for any choice of auxiliary system $\bm A$ (Hilbert space $\cH_{\bm A}$), its Hamiltonian $H_{\bm A}$ and the initial state $\rho_{\bm S \bm A}$. Clearly, we want to select them in a way to maximize the rhs in the above bound. With this in mind we  decompose the rest of the proof in  the following steps:

\begin{itemize}[leftmargin=46 pt]
    \item[\bf Step 1.] 
    We obtain a lower bound on the resource production $$ C\big(\nu_{\bm S\bm A}^{(m)}, H_{\bm S}+H_{\bm A}\big) 
 - C\big(\rho_{\bm S\bm A}^{(m)}, H_{\bm S}+H_{\bm A} \big)\geq (\ast) $$ in Eq.~\eqref{eq: term to LB}, as a function of the properties of the ideal gate $\mathcal{V}_S$. 
 A convenient choice is to restrict the analysis to initial states that are pure and eigenstates of the total Hamiltonian. In this way $C\big(\rho_{\bm S\bm A}^{(m)}, H_{\bm S}+H_{\bm A} \big)=0$ and 
 $C\big(\nu_{\bm S\bm A}^{(m)}, H_{\bm S}+H_{\bm A}\big)$ is mapped into the entropy of a sum of i.i.d. random variables (see Lemma \ref{Def:Energy-Prob-vector-lemma}).  Furthermore, the auxiliary systems are chosen to be copies of $S$, but with a flipped Hamiltonian $H_{A_i}=-H_S$.
    \item[\bf Step 2.] We obtain an upper bound on the resource regularity 
    $$C\big(\nu_{\bm S\bm A}^{(m)}, H_{\bm S}+H_{\bm A}\big) -C\big(\tnu_{\bm S\bm A}^{(m)}, H_{\bm S}+H_{\bm A}\big) \leq  (\ast\ast) $$ in Eq.~\eqref{eq: term to UB}
    as a function of the single-channel error $\epsilon_{wc}(\tcV_{S};\mathcal{V}_{S})$. Here, the idea consists of bounding the term with the trace distance between the states $\nu_{\bm S\bm A}^{(m)}$ and $\tnu_{\bm S\bm A}^{(m)}$, and then applying the bound \eqref{Th3}. 
     For this purpose, we will use a refinement of the Fannes-Audenaert inequality (see Lemma \ref{lem:ref_FanAud}) which holds in full generality. We then derive a specific bound for the choices of $H_{\bm A}$ and $\rho_{\bm S\bm A}^{(m)}$ used in step 1.
    We derive a general bound, and a refined one assuming that the battery is \textit{proportionate}.

    \item[\bf Step 3.] 
  { We combine the two bounds obtained in steps 1 and 2, which still depend on the choice of the initial pure state $\rho^{(m)}_{\bm S\bm A}$. We then select this state so that the difference $(\ast) - (\ast\ast) \leq C(\beta_B, H_{B})$, which lower bounds our quantity of interest via Eq.~\eqref{eq: LB-UB}, is provably large.}  Finally, we maximize the obtained bound with respect to the number of copies $m$. 
\end{itemize}
 
Step 1 of the proof holds for any resource $C$ fulfilling properties (C1-C4). In contrast, in the following we explicitly consider the entropic coherence of energy. 

\subsection{Step 1: Bounding the resource production} \label{App:BoundingProduction}

Our goal here is to choose specific initial states $\rho_{\bm S\bm A}^{(m)}$ that yield a convenient lower bound on the expression in Eq.~\eqref{eq: term to LB}, which corresponds to the amount of resource 
produced by the gate $\mathcal{V}_{\bm S}^{(m)}$ acting on the state. Using the definition of the entropic coherence we get 
\begin{align}\label{eq: rec to ent-0}
     C\big(\nu_{\bm S\bm A}^{(m)}, H_{\bm S}+H_A\big) 
 - C\big(\rho_{\bm S\bm A}^{(m)}, H_{\bm S}+H_{\bm A} \big) =   S\left(\mathcal{G}_{H_{\bm S \bm A}}[\nu_{\bm S\bm A}^{(m)}]\right) -S(\nu_{\bm S \bm A}^{(m)})  - S  \left(\mathcal{G}_{H_{\bm S \bm A}}[\rho_{\bm S\bm A}^{(m)}]\right) + S(\rho_{\bm S\bm A}^{(m)}).
\end{align}
Since the states are assumed pure and eigenstates of the total Hamiltonian, we have $S\left(\rho_{\bm S\bm A}^{(m)}\right)=S\left(\nu_{\bm S\bm A}^{(m)}\right)=S\left(\mathcal{G}_{H_{\bm S \bm A}}[\rho_{\bm S\bm A}^{(m)}]\right)=0$.
Hence, the `production of resource' is equal to the only remaining term, i.e. the entropy of the final state after twirling 
\begin{align}\label{eq: rec to ent}
     C\big(\nu_{\bm S\bm A}^{(m)}, H_{\bm S}+H_A\big) 
 - C\big(\rho_{\bm S\bm A}^{(m)}, H_{\bm S}+H_{\bm A} \big) = S\left(\mathcal{G}_{H_{\bm S \bm A}}[\nu_{\bm S\bm A}^{(m)}]\right).
\end{align}
To lower bound the right hand side it is convenient to identify it with the entropy of a random variable, given by the following definition.
\begin{defn} \label{Def:Energy-Prob-vector}
     Let $\psi=\ketbra{\psi}$ be a pure state of a $d$-dimensional quantum system, and $H= \sumab{j=1}{k}E_{j} \Pi_{E_{j}}$ the associated Hamiltonian. We call $X_{\psi}$ the discrete random variable taking values in the set
     \begin{equation} 
     \chi_{\psi} :=\{ E_{j}\in  \mathcal{E}(H): \bra{\psi}\Pi_{E_j}\ket{\psi} >0 \},
     \end{equation}
     and distributed according to ${\rm Pr}(X_\psi = E_j) = \bra{\psi}\Pi_{E_j}\ket{\psi}$, i.e. the random variable describing the outcome of an energy measurement on the state $\psi$. It is immediate to see that 
     \begin{equation}\label{eq: twirl = rv}
    S\left(\mathcal{G}_H[ \psi ] \right) = S(X_\psi).
\end{equation}
\end{defn}

The following lemma will also be useful.
\begin{lem} \label{Def:Energy-Prob-vector-lemma}
     Let $\Psi_{\bm T}$ be a pure state of  systems $\bm T = T_1 \dots T_m$, with associated Hamiltonians $H_{T_i}= \suma{j} E_i^{(j)} \Pi_{T_i}^{(j)}$  and $H_{\bm T} = \sumab{i=1}{m} H_{T_i}$. Then  
     \begin{equation}
         S\left(\mathcal{G}_{H_{\bm T}}[ \Psi_{\bm T} ] \right) = S\left(\sum_{i=1}^m X_{\Psi}^{(i)}\right),
     \end{equation}
where the random variables $X_\Psi^{(i)}$ take the values in $\cE(H_{T_i})$ and are distributed accordingly to $
    {\rm Pr}\Big(X_\Psi^{(1)} = E_1^{(j_1)},\dots,  X_\Psi^{(m)} = E_m^{(j_m)}\Big) = \bra{\Psi} \bigotimes_{i=1}^m  \Pi_{T_i}^{(j_i)} \ket{\Psi}.$
In addition, if the state is a product, $\Psi_{\bm T}=\bigotimes_{i=1}^m \ketbra{\psi_i}_{T_i}$, the random variables are independent and given by $X_\Psi^{(i)}= X_{\psi_i}$.
\end{lem}
\begin{proof} By definition the random variable $X_\Psi:=\sumab{i=1}{m} X_{\Psi}^{(i)}$ takes values in $\cE(H_{\bm T})$ and is distributed accordingly to 
\begin{equation}
    {\rm Pr}(X_\Psi = E) = \bra{\Psi} \Pi_E \ket{\Psi}, \qquad \text{where} \qquad \Pi_E = \sum_{j_1,\dots j_m | \sum_i E_{i}^{(j_i)} =E} \quad \bigotimes_{i = 1}^m \Pi_{T_i}^{(j_i)}
\end{equation}
is the projector on the subspace with total energy $E$, i.e. $H_{\bm T} = \suma{E\in \cE(H_{\bm T})} E \, \Pi_E$. Hence, by Eq.~\eqref{eq: twirl = rv} we obtain 
 \begin{equation}
    S\left(\mathcal{G}_{H_{\bm T}}[ \Psi_{\bm T} ] \right) = S(X_\Psi)= S\left(\sum_{i=1}^m X_{\Psi}^{(i)}\right).
\end{equation}

To see that for a product initial state $\Psi_{\bm T}=\bigotimes_{i=1}^m \ketbra{\psi_i}_{T_i}$ the random variables are independent and satisfy $X_\Psi^{(i)}= X_{\psi_i}$, remark that
    \begin{equation}
    {\rm Pr}\Big(X_\Psi^{(1)} = E_1^{(j_1)},\dots,  X_\Psi^{(m)} = E_m^{(j_m)}\Big) =  \bra{\Psi} \bigotimes_{i=1}^m  \Pi_{T_i}^{(j_i)} \ket{\Psi}= \prod_{i=1}^m \bra{\psi_i}   \Pi_{T_i}^{(j_i)} \ket{\psi_i} =  \prod_{i=1}^m {\rm Pr}\left(X_{\psi_i} = E_i^{(j_i)}\right).
\end{equation}
\end{proof}

To proceed further we consider separately the general case and the case of a single qubit, where the derived bound is tighter. We start with the general case.

\subsubsection{Resource production for general gates}

For the general case of a $d$-dimensional system, we will take the initial state $\psi_{\bm S \bm A}$ to be a product. Specifically we now group the subsystems $S_1,\dots S_{2m}$ and $A_1,\dots A_{2m}$ into $m$ subsystems
$$\bm T = \bm S \bm A = T_1 \dots T_m \qquad \text{with} \qquad T_i = S_{2i-1}S_{2i}A_{2i-1}A_{2i},$$
and restrict the initial state to be a product of identical states
\begin{equation}\label{app: rho in d}
    \rho_{\bm T}^{(m)} = \rho_{\bm S\bm A}^{(m)} = \bigotimes_{i=1}^m \ketbra{\varphi}_{T_i} \qquad \text{with} \qquad 
    \ket{\varphi}_{T_i} \in \aut[H_{S_{2i-1}}+H_{S_{2i}}+H_{A_{2i-1}}+H_{A_{2i}}].
\end{equation}
After the application of the gate the state becomes
\begin{equation}
\nu_{\bm T}^{(m)} =  \mathcal{V}_{\bm S}^{(m)}[\rho_{\bm S\bm A}^{(m)}]= \bigotimes_{i=1}^m \ketbra{\psi }_{T_i} \qquad \text{with identical} \qquad  \ket{\psi}_T=(V_{S_{1}}\otimes V_{S_{2}}^\dag\otimes \id_{A_{1}A_{2}}) \ket{\varphi}_T.
\end{equation}

By Lemma~\ref{Def:Energy-Prob-vector-lemma} we conclude that the entropy of the final state after twirling is given by 
\begin{equation}
 S\left(\mathcal{G}_{H_{\bm S \bm A}}[\nu_{\bm S\bm A}^{(m)}]\right) = 
S\left(\sum_{i=1}^m X_\psi^{(i)} \right),
\end{equation}
where $X_{\psi}^{(i)}$ are i.i.d. random variables describing the energy measurement of the state $\ket{\psi}_T$ and taking values in the set $\chi_\psi$. Combining with Eq.~\eqref{eq: rec to ent} and using the lower bound on the entropy of the sum of i.i.d. random variables (Corollary \ref{cor:GeneralIIDEntropy}) derived in the dedicated Appendix~\ref{app: iid}, we obtain the following bound on the resource production term  
\begin{equation}\label{app: bound prod 1}
      C\big(\nu_{\bm S\bm A}^{(m)}, H_{\bm S}+H_A\big) 
 - C\big(\rho_{\bm S\bm A}^{(m)}, H_{\bm S}+H_{\bm A} \big) = S\left(\sum_{i=1}^m X_\psi^{(i)} \right)\geq  \frac{r(\chi_{\psi})}{2}\log\big(2 \pi e \,m \, \lambda(X_{\psi})\big )-o(1),
\end{equation}
where $X_\psi$ is the energy random variable of Definition~\ref{Def:Energy-Prob-vector}, $r(\chi_\psi)$, given in Definition~\ref{Def:r-Set}, depends on the spectrum of $X_\psi$, $\lambda(X_\psi)$, defined in Eq.~\eqref{eq: lambda1}, also depends on its distribution, and $o(1)$ refers to the $m \rightarrow \infty$ limit.\\

Next, we show that for any NEPG $V_S$ it is possible to choose the initial state $\ket{\varphi}_T$ such that $r(\chi_\psi)\geq 1$ and $\lambda(X_\psi)>0$, which (by definition of these quantities) is equivalent to the random variable $X_\psi$ taking at least two different values ($|\chi_{\psi}|\geq 2$). This is summarized by the following lemma

\begin{lem} \label{cor:r-non-zero}
 If $[V_{S},H_{S}]\neq 0$ then for four copies of $S$ denoted $T= S_1S_2A_1A_2$ with Hamiltonians $H_T=H_{S_1}+H_{S_2}+H_{A_1}+H_{A_2}$ and $H_{S_i}= -H_{A_i}=H_{S}$, there are states
 \begin{equation}
     \ket{\varphi}_T \in \aut[H_T] \qquad \text{and} \qquad \ket{\psi}_T = V_{S_1}\otimes V_{S_2}^{\dag}\otimes \id_{A_1A_2}\ket{\varphi}_T 
 \end{equation}
 such that $|\chi_\psi|\geq 2$, i.e. $\ket{\psi}_T$ is not an eigenstate of $H_T$.
\end{lem}

\begin{proof}
Set
\begin{equation}
    \overline{H}:=H_{S_1}+H_{S_2},
    \qquad
    W:=V_{S_1}\otimes V_{S_2}^{\dagger}.
\end{equation}
We first show that $[W,\overline{H}]\neq0$. Indeed, defining
\begin{equation}
    O_{S}:=V_S^{\dagger}H_SV_S-H_S,
\end{equation}
the assumption $[W,\overline{H}]=0$ would imply
\begin{equation}
    0=W^{\dagger}\overline{H}W-\overline{H}
    =O_{S}\otimes\id_S-\id_S\otimes V_SO_{S} V_S^{\dagger}.
\end{equation}
Taking the partial trace over the second system and using $\tr O_{S}=0$, we obtain
\begin{equation}
    0=d_S O_{S}-\tr(O_{S})\id_S=d_S O_{S}.
\end{equation}
Thus $O_{S}=0$, which is equivalent to $[V_S,H_S]=0$, contradicting the hypothesis.

Let $D=d_S^2$ and let $\{\ket{E_\alpha}_{S_1S_2}\}_{\alpha=1}^{D}$ be an orthonormal eigenbasis of $\overline{H}$. For each $\alpha$, let $\ket{-E_\alpha}_{A_1A_2}$ denote the corresponding eigenstate of
\begin{equation}
    H_A:=H_{A_1}+H_{A_2}
\end{equation}
with eigenvalue $-E_\alpha$, and define
\begin{equation}
    \ket{\varphi}_T
    :=
    \frac{1}{\sqrt{D}}
    \sum_{\alpha=1}^{D}
    \ket{E_\alpha}_{S_1S_2}
    \ket{-E_\alpha}_{A_1A_2}.
\end{equation}
Clearly, $\ket{\varphi}_T\in\aut[H_T]$ with total energy zero. Set
\begin{equation}
    \ket{\psi}_T
    :=
    \left(W\otimes\id_{A_1A_2}\right)\ket{\varphi}_T.
\end{equation}
Its average total energy is
\begin{equation}
    \bra{\psi}H_T\ket{\psi}
    =
    \frac{1}{D}\tr\!\left(W^{\dagger}\overline{H}W-\overline{H}\right)
    =0.
\end{equation}
Therefore, if $\ket{\psi}_T$ were an eigenstate of $H_T$, its eigenvalue would necessarily be zero. However,
\begin{equation}
    H_T\ket{\psi}_T
    =
    \left([\overline{H},W]\otimes\id_{A_1A_2}\right)
    \ket{\varphi}_T.
\end{equation}
Since $\ket{\varphi}_T$ has full Schmidt rank across
$S_1S_2:A_1A_2$, the right-hand side vanishes if and only if
$[\overline{H},W]=0$, which we have already excluded. Hence
$\ket{\psi}_T$ is not an eigenstate of $H_T$, and therefore
$|\chi_\psi|\geq2$.
\end{proof}

The lemma guarantees that the bound~\eqref{app: bound prod 1} is nontrivial for any NEPG. Nevertheless, we want to make it as tight as possible. To do so we introduce a formal maximization of the rhs of Eq.~\eqref{app: bound prod 1} with respect to the choice of the initial state. Since $r(\chi_\psi)$ enters linearly in the bound, and $\lambda(X_\psi)$ logarithmically, we are primarily interested in maximizing the former. This gives rise to the following definitions.
\begin{defn}\label{app: defn r lambda new}
    Consider a Hamiltonian $H_S$ and a NEPG $V_S$. Let $T=S_1S_2 A_1 A_2$ be composed of four copies of the system $S$ with $H_S=H_{S_i}= -H_{A_i}$ and $H_T=H_{S_1}+H_{S_2}+H_{A_1}+H_{A_2}$. Define
    \begin{align}
        r_2(V_S,H_S) = \max \quad &r(\chi_\psi) \\
        {\rm such \, that }\quad &\ket{\varphi}_T\in \aut[H_T]
        \\ 
        &\ket{\psi}_T = V_{S_1}\otimes V_{S_2}^\dag \otimes \id_{A_1A_2} \ket{\varphi}_T
    \end{align}
    where $\chi_\psi$ is the energy spectrum of the state $\ket{\psi}_T$ (see  Def.~\ref{Def:Energy-Prob-vector}),  and $r(\chi_\psi)$ is given in the definition~\ref{Def:r-Set}. In addition, define
    \begin{align}
        \lambda_2(V_S,H_S) = \max \quad & \lambda(X_\psi) \\
        {\rm such \, that }\quad &\ket{\varphi}_T\in \aut[H_T] \\
        &\ket{\psi}_T = V_{S_1}\otimes V_{S_2}^\dag \otimes \id_{A_1A_2} \ket{\varphi}_T \\
        & r(\chi_\psi) = r_2(V_S,H_S) 
    \end{align}
    where $X_\psi$ is the energy distribution of the state $\ket{\psi}_T$ (see  Def.~\ref{Def:Energy-Prob-vector}), and $\lambda(X_\psi)$ is given in Eq.~\eqref{eq: lambda1}.
\end{defn}

Computing $\lambda_2(V_{S},H_S)$ explicitly may be difficult, however lower bounds are very easy to compute for example using the state guessed in Lemma \ref{cor:r-non-zero}.  Combining everything we summarize the results of this section in the following proposition.

\begin{rslt}[Entropic coherence production  for a general gate]  \label{lem:Entropy-twirled}
For a $d$-dimensional quantum system with Hamiltonian $H_S$, and a NEPG $V_S$ there is a choice of pure initial state $\rho_{\bm S \bm A}$, which is an eigenstate of the total energy $H_{\bm S \bm A}$, such that the resource production term satisfies 
\begin{equation}\label{eq: bound r prod gen}
      C\big(\nu_{\bm S\bm A}^{(m)}, H_{\bm S}+H_{\bm A}\big) 
 - C\big(\rho_{\bm S\bm A}^{(m)}, H_{\bm S}+H_{\bm A} \big) \geq  \frac{r_2(V_S,H_S)}{2} \log\big(2 \pi e \,m \, \lambda_2(V_S,H_S)\big)-o(1),
\end{equation}
where $r_2(V_S,H_S) \in [1,\dots,d^{2}-1]$ and $\lambda_2(V_S,H_S)> 0$ are given in the definition~\ref{app: defn r lambda new} and $o(1)$ refers to the $m \rightarrow \infty$ limit.
\end{rslt}

\begin{proof}
To obtain the bound we simply combined 
the definitions~\ref{app: defn r lambda new} with the bound on the entropy in Eq.~\eqref{app: bound prod 1}. The bounds $r_2(V_S,H_S)\geq 1$ and $\lambda_2(V_S,H_S)> 0$ follow from the lemma~\ref{cor:r-non-zero}, while $r_2(V_S,H_S)\leq d^{2}-1$ follows from the definition of $r$.
\end{proof}

Finally we prove a simple lower bound on $r_{2}(V_{S},H_{S})$ that holds for \textit{generic} $V_{S},H_{S}$. We phrase the statement choosing specific measures for unitaries and Hamiltonians, but, as can be readily verified, the specific choice is not crucial. 
\begin{lem}
    Let $V_{S},H_{S}$ be respectively sampled from the Haar measure \cite{HaarTutorial} and the Gaussian Unitary ensemble \cite{GUEBook} on the $d_{S}$-dimensional Hilbert space (or from absolutely continuous measures with respect to them). Then, with probability one
    \begin{equation}
        r_{2}(V_{S},H_{S})\geq d_{S}-1.
    \end{equation}
\end{lem}

\begin{proof}   
   Following the Def.~\ref{app: defn r lambda new} consider the initial state 
    \begin{equation}
    \ket{\varphi}_{T}:=\ket{E_{1},E_{1}}_{S_{1}S_{2}}\otimes \ket{E_{1},E_{1}}_{A_{1}A_{2}}\in \aut[H_T].
    \end{equation}
By definition we have $r_2(V_S,H_S)\geq r(\chi_\psi)$ where $\ket{\psi}_T= V_{S_{1}}\otimes V^{\dagger}_{S_{2}} \otimes \id_{A_1A_2}\ket{\varphi}_T$ and $\chi_\psi$ is its energy spectrum. With probability one, the random gate $V_S$ has no zeros (when written in the energy eigenbasis), implying 
\begin{equation}
{}_{S_{1}S_{2}}\!\bra{E_i,E_j}\,{}_{A_{1}A_{2}}\!\bra{E_1,E_1}\ket{\psi}_T\neq 0 \qquad \forall i,j =1,\dots, d_S
\end{equation}
Hence, almost surely, the target set contains all sums of energies of the form $\chi_\psi=\{E_i+E_j-2E_1\}_{i,j}$, and 
\begin{equation}
r_2(V_S,H_S)\geq r(\{E_i+E_j-2E_1\}_{i,j})= r(\{E_i+E_j\}_{i,j}),
\end{equation}
where in the last equality we used the fact that the incommensurability rank of a set is invariant under shifts. Indeed, for any $a\in\mathbb{R}$, two sums of $N$ elements of a set $\chi$ coincide if and only if the corresponding sums of elements of $\chi+a$ coincide, since both sums are shifted by $Na$. Hence, shifting the set leaves unchanged the additive relations entering the definition of a canonical prepartition, and therefore $r(\chi+a)=r(\chi)$.

Next, we recall that by Lemma \ref{lem: part span}, $r(\{E_i+E_j\}_{i,j})\geq {\rm dim}_{\mathbb{Q}}[{\rm span}_{\mathbb{Q}}[\{E_i+E_j\}_{i,j}]]-1$, where ${\rm span}_{\mathbb{Q}}[\chi]= \sum_{i=1}^{|\chi|} q_i x_i$ with $q_i \in \mathbb{Q}$ and $x_i \in \chi$ is the rational span of a set. Furthermore, it is easy to see that ${\rm span}_{\mathbb{Q}}[\{E_i+E_j\}_{i,j}]= {\rm span}_{\mathbb{Q}}[\{E_i\}_{i}]$, leading to 
\begin{equation}
  r_2(V_S,H_S) \geq {\rm dim}_{\mathbb{Q}}[{\rm span}_{\mathbb{Q}}[\{E_i\}_{i}]]-1.
\end{equation}

Finally, note that for a random Hamiltonian $H_S$ its $d_{S}$ eigenvalues $E_i$ are real numbers that are almost surely linearly independent over $\mathbb{Q}$, i.e. for $q_{i}\in \mathbb{Q}$ we have $\sumab{i=1}{d_{S}}q_{i}E_{i}=0$ if and only if $q_{i}=0$. Hence, with probability one $ {\rm dim}_{\mathbb{Q}}[{\rm span}_{\mathbb{Q}}[\{E_i\}_{i}]]=d_S$, completing the proof.
\end{proof}

\subsubsection{Resource production for a qubit gate}

We now specifically consider the case of a qubit gate $\cV_S$, with $\ket{0}_S$ and $\ket{1}_S$ denoting the two eigenstates of $H_S$ with different energies. The gate can then be represented by the matrix
\begin{equation}\label{app: gate qubit}
    V_S = \left(\begin{array}{cc}
        \cos(\theta) & \sin(\theta) e^{\ii \varphi} \\
         -\sin(\theta)e^{-\ii \varphi} &\cos(\theta)
    \end{array}\right)
\end{equation}
in the energy basis, and for the purpose of this section we can assume that the energies are $E_0=0$ and $E_1=1$ without loss of generality. 

Let us choose the pure initial state $\rho_{\bm S\bm A}^{(m)} =\ketbra{\psi^{(m)}}_{\bm S\bm A}$ to be given by
\begin{equation} \label{Eq:DefInitialQubit}
\ket{\psi^{(m)}}_{\bm S\bm A} =  \frac{1}{\sqrt{2m+1}}\sum_{k=0}^{2m}\bigotimes_{i=1}^k \ket{1,1}_{S_iA_i} \bigotimes_{i=k+1}^{2m} \ket{0,0}_{S_iA_i}  =\frac{1}{\sqrt{2m+1}}\sum_{k=0}^{2m}\ket{1^{\otimes k},0^{\otimes (2m-k)}}_{\bm S}\otimes \ket{1^{\otimes k},0^{\otimes (2m-k)}}_{\bm A} ,  
\end{equation}
where $\ket{1^{\otimes k},0^{\otimes 2m-k}}_{\bm S}$ means that the systems $S_1,\dots, S_k$ are in the state $\ket{1}_{S_i}$ and the remaining systems are in the state $\ket{0}_{S_i}$. Note that since all $A$ are copies of $S$ with inverted energy scale  we have $(H_{S_i}+ H_{A_i})\ket{1,1}_{S_iA_i} =(H_{S_i}+ H_{A_i})\ket{0,0}_{S_iA_i}=0$, and $\ket{\psi^{(m)}}_{\bm S\bm A}$ is an eigenstate of $H_{\bm S \bm A}$ of total energy zero.

The $2m$-copy  gate of interest can be written as $\cV_{\bm S}^{(m)}=\bigotimes_{j=1}^{m}(\cV_{S_{2j-1}}\otimes \cV_{S_{2j}}^*) = \bigotimes_{i=1}^{2m} \bar \cV_{\bm S_i}$ where slightly abusing the notation we introduced $ \bar \cV_{S_i} = \cV_{S_i}$ for odd $i$ and $ \bar \cV_{S_i} = \cV_{S_i}^*$ for even $i$. After its application the state becomes
\begin{align} 
\ket{\nu^{(m)}}_{\bm S\bm A} \coloneqq \cV_{\bm S}^{(m)} \ket{\psi^{(m)}}_{\bm S\bm A} = \frac{1}{\sqrt{2m+1}}\sum_{k=0}^{2m} \ket{\zeta_k^{(m)} }_{\bm S \bm A} \qquad \text{with} \label{Eq:DefFinalQubit} \\
\ket{\zeta_k^{(m)} }_{\bm S \bm A}:= \bigotimes_{i=1}^k \bar V_{S_i} \ket{1}_{S_i} \bigotimes_{i=k+1}^{2m} \bar V_{S_i} \ket{0}_{S_i} \otimes \ket{1^{\otimes k},0^{\otimes (2m-k)}}_{\bm A}.
\end{align}

\begin{rslt} \label{res: res prod qubit}
[Entropic coherence production for a qubit gate]
For a qubit NEPG $V_S$ in Eq.~\eqref{app: gate qubit},  let $\rho_{\bm S\bm A}^{(m)}$ and $\nu^{(m)}_{\bm S\bm A}$ be defined as in Eqs.\eqref{Eq:DefInitialQubit},\eqref{Eq:DefFinalQubit}
then  in the $m\rightarrow \infty$ limit the following bound holds
\begin{equation}
       C\big(\nu_{\bm S\bm A}^{(m)}, H_{\bm S}+H_A\big) 
 - C\big(\rho_{\bm S\bm A}^{(m)}, H_{\bm S}+H_{\bm A} \big)\geq \log(4m\sin^2(\theta))-\delta_{1,\sin^2(\theta)}-o(1).
  \end{equation}    
Note that this bound can be put in the form of the general Result~\ref{lem:Entropy-twirled} by setting
\begin{equation}
r_2(V_S,H_S)=2 \qquad \text{and}\qquad \lambda_2(V_{S},H_S) =  \frac{2 \sin^2(\theta)}{\pi e (1+\delta_{1,\sin^2(\theta)})}
\end{equation}
in Eq.~\eqref{eq: bound r prod gen} for all qubit gates.
\end{rslt}

\begin{proof}
Set
\begin{equation}
    p:=\sin^2(\theta),\qquad q:=1-p,\qquad N:=2m+1.
\end{equation}
Since $\rho_{\bm S\bm A}^{(m)}$ is a pure state supported on the
zero-energy eigenspace,
\begin{equation}
    C\big(\rho_{\bm S\bm A}^{(m)},H_{\bm S}+H_{\bm A}\big)=0.
\end{equation}
Moreover, since $\nu_{\bm S\bm A}^{(m)}$ is pure,
\begin{equation}
    C\big(\nu_{\bm S\bm A}^{(m)},H_{\bm S}+H_{\bm A}\big)
    =S\big(X_{\nu^{(m)}}\big).
\end{equation}

The ancillary states associated with different values of $k$ are
orthogonal energy eigenstates. Therefore,
\begin{equation}
    X_{\nu^{(m)}}[n]
    =
    \frac{1}{2m+1}
    \sum_{k=0}^{2m}X_{\zeta_k^{(m)}}[n].
\end{equation}
For fixed $k$, Eq.~\eqref{Eq:DefFinalQubit} gives
\begin{equation}
    X_{\zeta_k^{(m)}}
    \overset{\rm d}{=}
    R_{2m-k}-L_k,
\end{equation}
where $R_{2m-k}$ and $L_k$ are independent binomial random variables
with parameters $(2m-k,p)$ and $(k,p)$, respectively.

Consider the Laurent probability-generating function
\begin{equation}
    G_m(z):=
    \sum_{n=-2m}^{2m}X_{\nu^{(m)}}[n]z^n.
\end{equation}
We have
\begin{align}
    G_m(z)
    &=
    \frac{1}{N}
    \sum_{k=0}^{N-1}
    (q+pz)^{N-1-k}(q+pz^{-1})^k \\
    &=
    \frac{(q+pz)^N-(q+pz^{-1})^N}
    {Np(z-z^{-1})} \\
    &=
    \frac{1}{Np}
    \sum_{j=1}^{N}
    \binom{N}{j}q^{N-j}p^j
    \sum_{\ell=0}^{j-1}z^{j-1-2\ell}.
\end{align}
It follows that
\begin{equation}
    X_{\nu^{(m)}}[n]
    =
    \frac{1}{Np}
    \sum_{\substack{j\geq |n|+1\\
    j\equiv n+1\ ({\rm mod}\,2)}}
    \binom{N}{j}p^j q^{N-j}.
\end{equation}
Equivalently, letting $J\sim{\rm Bin}(N,p)$,
\begin{align}
    X_{\nu^{(m)}}[n]
    &=
    \frac{1}{Np}
    {\rm Pr}\!\left(
    J\geq |n|+1,\,
    J\equiv n+1\!\!\pmod 2
    \right) \\
    &\leq
    \frac{1}{Np}
    {\rm Pr}\!\left(J\equiv n+1\!\!\pmod 2\right).
\end{align}
The parity probabilities of a binomial random variable satisfy
\begin{equation}
    {\rm Pr}(J\equiv r\!\!\pmod 2)
    =
    \frac{1+(-1)^r(q-p)^N}{2},
\end{equation}
and hence
\begin{equation}
    \max_nX_{\nu^{(m)}}[n]
    \leq
    \frac{1+|1-2p|^N}{2Np}.
\end{equation}
Using the min-entropy bound on the Shannon entropy, we obtain
\begin{align}
    S\big(X_{\nu^{(m)}}\big)
    &\geq
    -\log\max_nX_{\nu^{(m)}}[n] \\
    &\geq
    \log(2Np)-\log\big(1+|1-2p|^N\big).
\end{align}

For fixed $0<p<1$, this gives
\begin{equation}
    S\big(X_{\nu^{(m)}}\big)
    \geq
    \log(4mp)-o(1).
\end{equation}
For $p=1$, it gives
\begin{equation}
    S\big(X_{\nu^{(m)}}\big)
    \geq
    \log(2m+1)
    =
    \log(4m)-1+o(1).
\end{equation}
Recalling that $p=\sin^2(\theta)$ completes the proof.
\end{proof}

\subsection{Step 2: Refining resource regularity}

The next step is to upper bound the difference in the relative entropy produced by the ideal and approximate gates
\begin{align}C( {\nu}^{(m)}_{\bm S\bm A},H_{\bm S\bm A}) - C({\tnu}^{(m)}_{\bm S\bm A},H_{\bm S \bm A}) &= S(\mathcal{G}_{H_{\bm S \bm A}}[{\nu}^{(m)}_{\bm S\bm A}]) -S(\mathcal{G}_{H_{\bm S \bm A}}[{\tnu}^{(m)}_{\bm S\bm A}]) \label{eq: FA1} 
+ S({\tnu}^{(m)}_{\bm S\bm A}) - S({\nu}^{(m)}_{\bm S\bm A}), 
\end{align}
as a function of the trace distances between the final states  $D\left( \tnu_{\bm S \bm A}^{(m)},\nu_{\bm S\bm A}^{(m)} \right)$, which in turn is bounded by the gate error via Eq.~\eqref{eq:distance4m}. A key ingredient here is the following regularity condition for the von Neumann entropy.

\begin{lem}[Refined Fannes-Audenaert inequality  \cite{EntropyContinuity2024}] \label{lem:ref_FanAud}
    For density operators $\rho$ and $\sigma$ one has
\begin{equation}
    S(\rho)-S(\sigma)\leq D(\rho,\sigma)(S(\Delta_{+})-S(\Delta_{-}))+h_2\big(D(\rho,\sigma)\big)
\end{equation}
where $\Delta_{+/-}$ are the unique positive, orthogonal operators satisfying $D(\rho,\sigma)(\Delta_{+}-\Delta_{-})=\rho-\sigma$ (Jordan-Hahn decomposition).
\end{lem}
Since ${\rm Rank}(\rho)\geq {\rm Rank}\Delta_{+}$, we also have 
\begin{equation}
    S(\rho) - S(\sigma) \leq \log\big({\rm Rank}(\rho)\big)D(\rho,\sigma)+h_2\big(D(\rho,\sigma)\big),
\end{equation}
Applying this bound to the entropy differences in Eq.~\eqref{eq: FA1} we obtain
\begin{align}
S({\tnu}^{(m)}_{\bm S\bm A}) - S({\nu}^{(m)}_{\bm S\bm A}) &\leq \log\big({\rm Rank}({\tnu}^{(m)}_{\bm S\bm A})\big) D({\tnu}^{(m)}_{\bm S\bm A},{\nu}^{(m)}_{\bm S\bm A}) + h_2\big(D({\tnu}^{(m)}_{\bm S\bm A},{\nu}^{(m)}_{\bm S\bm A})\big)
\\
    S(\mathcal{G}_{H_{\bm S \bm A}}[{\nu}^{(m)}_{\bm S\bm A}]) -S(\mathcal{G}_{H_{\bm S \bm A}}[{\tnu}^{(m)}_{\bm S\bm A}]) &\leq \log\big({\rm Rank}(\mathcal{G}_{H_{\bm S \bm A}}[{\nu}^{(m)}_{\bm S\bm A}])\big) D(\mathcal{G}_{H_{\bm S \bm A}}[{\nu}^{(m)}_{\bm S\bm A}],\mathcal{G}_{H_{\bm S \bm A}}[{\tnu}^{(m)}_{\bm S\bm A}]) \\ &+ h_2\big(D(\mathcal{G}_{H_{\bm S \bm A}}[{\nu}^{(m)}_{\bm S\bm A}],\mathcal{G}_{H_{\bm S \bm A}}[{\tnu}^{(m)}_{\bm S\bm A}])\big)\\
    &\leq \log\big({\rm Rank}(\mathcal{G}_{H_{\bm S \bm A}}[{\nu}^{(m)}_{\bm S\bm A}])\big)  D({\tnu}^{(m)}_{\bm S\bm A},{\nu}^{(m)}_{\bm S\bm A}) \\
    &+h_2\big(D({\tnu}^{(m)}_{\bm S\bm A},{\nu}^{(m)}_{\bm S\bm A})\big),
\end{align}
where we used the monotonicity of the trace distance $D(\mathcal{G}_{H_{\bm S \bm A}}[{\nu}^{(m)}_{\bm S\bm A}],\mathcal{G}_{H_{\bm S \bm A}}[{\tnu}^{(m)}_{\bm S\bm A}])  \leq D({\tnu}^{(m)}_{\bm S\bm A},{\nu}^{(m)}_{\bm S\bm A})$, and the fact that the binary entropy $h_2(x)$ is an increasing function for $x\leq \frac{1}{2}$ (assuming $D({\tnu}^{(m)}_{\bm S\bm A},{\nu}^{(m)}_{\bm S\bm A})\leq \frac{1}{2}$). Combining the two bounds we obtain the following intermediate results
\begin{equation}\label{eq: taming bound}
C( {\nu}^{(m)}_{\bm S\bm A},H_{\bm S\bm A}) - C({\tnu}^{(m)}_{\bm S\bm A},H_{\bm S \bm A}) \leq \log\big({\rm Rank}({\tnu}^{(m)}_{\bm S\bm A})\, {\rm Rank}(\mathcal{G}_{H_{\bm S \bm A}}[{\nu}^{(m)}_{\bm S\bm A}])\big) D({\tnu}^{(m)}_{\bm S\bm A},{\nu}^{(m)}_{\bm S\bm A}) + 2\,  h_2\big(D({\tnu}^{(m)}_{\bm S\bm A},{\nu}^{(m)}_{\bm S\bm A})\big).
\end{equation}
The rest of the section is hence devoted to bound the ranks of ${\tnu}^{(m)}_{\bm S \bm A}$ and $\mathcal{G}_{H_{\bm S \bm A}}[{\nu}^{(m)}_{\bm S\bm A}]$.

\subsubsection{Bounding the rank of  $\mathcal{G}_{H_{\bm S \bm A}}[{\nu}^{(m)}_{\bm S\bm A}]$}

To bound the rank of $\mathcal{G}_{H_{\bm S \bm A}}[{\nu}^{(m)}_{\bm S\bm A}]$, notice that the state ${\nu}^{(m)}_{\bm S\bm A}$ is pure, therefore after twirling its rank is bounded by the number of different energies of the total Hamiltonian (number of projectors in the twirling map)
\begin{equation}
    {\rm Rank}(\mathcal{G}_{H_{\bm S \bm A}}[{\nu}^{(m)}_{\bm S\bm A}]) \leq \abs{\mathcal{E}(H_{\bm S \bm A})}.
\end{equation}
Recall, that the total system consists of $2m$ copies of $SA$ with Hamiltonian $H_{SA}=H_S +H_A$. Since $H_S$ and $H_A$ are identical up to the sign flip, one sees that $\abs{\mathcal{E}(H_{S A})}\leq d_S^2 - (d_S-1)$, since the energy $0$ is $d_S$-degenerate. Now, since all $2m$ copies of $SA$ have the same Hamiltonian, the total Hamiltonian $H_{\bm S\bm A}$ is invariant under their permutation. Every total energy can be written as a signed sum of the $\#=\frac{d_S(d_S-1)}{2}$ positive energy gaps, with non-negative multiplicities whose sum is at most $2m$. The number of such choices is at most the number of ways of choosing the multiplicities, together with one sign for each gap, i.e.
\begin{equation}\label{app eq: bound rank gen}
    {\rm Rank}(\mathcal{G}_{H_{\bm S \bm A}}[{\nu}^{(m)}_{\bm S\bm A}])
    \leq \abs{\mathcal{E}(H_{\bm S \bm A})}
    \leq 
    2^{\frac{d_S(d_S-1)}{2}}
    \frac{(2m+\#)!}{(2m)!(\#)!}
    =
    2^{\frac{d_S(d_S-1)}{2}}
    \binom{2m+\frac{d_S(d_S-1)}{2}}{2m}.
\end{equation}
Notice however that for $d_{S}=2$ we have $\abs{\mathcal{E}(H_{\bm S \bm A})}\leq 4m+1$.

\subsubsection{Bounding the rank of  ${\tnu}^{(m)}_{\bm S\bm A}$}

Next, to bound the rank of  ${\tnu}^{(m)}_{\bm S\bm A}$ we consider three different approaches, depending on what is known about the battery system. \\

{\it (i)} A straightforward but generic bound is given by the dimension of the Hilbert space on which the total system lives 
\begin{equation}\label{eq: bound rank no ass}
      {\rm Rank}({\tnu}^{(m)}_{\bm S\bm A}) \leq {\rm dim}( \mathcal{H}_{\bm S \bm A}) = d_S^{4m}.
\end{equation}

{\it (ii)}
Now we derive a strengthened bound on the rank of ${\tnu}^{(m)}_{\bm S\bm A}$ valid for a \textit{proportionate} battery.  For the sake of clarity we start by considering the simplified hypothesis of a fixed $d_{B}$-dimensional battery. Recall that the state of interest is obtained via ${\tnu}^{(m)}_{\bm S\bm A} = \tr_B \mathcal{U}_{\bm S B}^{(m)} [\rho^{(m)}_{\bm S \bm A} \otimes \beta_B]$ from  $\rho^{(m)}_{\bm S \bm A}=\ketbra{\psi^{(m)}}_{\bm S \bm A}$. Then, if the initial state of the battery is pure $\beta_B = \ketbra{\beta}_B$, the final state $U_{\bm S B}^{(m)} \ket{\psi^{(m)}}_{\bm S \bm A}\ket{\beta}_B$ is also pure. Hence, for the bipartition $\bm S \bm A |B$ it can be put in the Schmidt diagonal form with at most $d_B$ non-zero coefficients, ensuring that ${\rm Rank}(\tr_B \mathcal{U}_{\bm S B}^{(m)} [\rho^{(m)}_{\bm S \bm A} \otimes \beta_B])\leq d_B$. In turn, if $\beta_B$ is not pure, it is a mixture of ${\rm Rank}(\beta_B)\leq d_B$ pure states. Then ${\tnu}^{(m)}_{\bm S\bm A}$ is a mixture of at most ${\rm Rank}(\beta_B)$ states with rank at most $d_B$, leading to the following bound
\begin{equation}
      {\rm Rank}({\tnu}^{(m)}_{\bm S\bm A})  \leq {\rm Rank}(\beta_B) \, d_B \leq d_B^{2}.
\end{equation}

 Similarly, suppose that $\beta_B$ is supported on a subspace spanned by eigenstates of $H_B$ with energy range $[0,E_{max}(\beta_B)]$. Assume first that the battery $\beta_B=\ketbra{\beta_B}$ is pure. Since $\ket{\psi^{(m)}}_{\bm S \bm A}$ is an eigenstate of $H_{\bm S \bm A}$ with eigenvalue $E_\psi$, and $\mathcal{U}_{\bm S B}$ is energy preserving, the total energy of the final state 
\begin{equation}
\ket{\xi}_{\bm S\bm A B} =U_{\bm S B}^{(m)} \ket{\psi^{(m)}}_{\bm S \bm A}\ket{\beta}_B
\end{equation}
must be smaller than $E_{\rm max}(\beta_B)+E_\psi$. For any energy levels of the system $H_{\bm S\bm A} \ket{E}_{\bm S\bm A} = E \ket{E}_{\bm S\bm A}$ and the battery $H_B \ket{E'}_B =E' \ket{E'}_B$, we must thus have
\begin{equation}
    \bra{E}_{\bm S\bm A}\bra{E'}_B \ket{\xi}_{\bm S\bm A B} \neq 0 \implies E+ E' \leq E_{\rm max}(\beta_B)+E_\psi.
\end{equation}
Denoting with $E^{(m)}_{\min}$ the minimum  eigenvalue of $H_{\bm S \bm A}$ we conclude that the final energy of the battery system must satisfy
\begin{equation}
    \bra{E'}_B \ket{\xi}_{\bm S\bm A B} \neq \bm 0 \implies E' \leq  E_{\rm max}(\beta_B)+E_\psi - E_{\min}^{(m)}.
\end{equation}
It follows that
\begin{align}
    {\rm Rank}({\tnu}^{(m)}_{\bm S\bm A}) = {\rm Rank}(\tr_{\bm S \bm A} \ketbra{\xi}_{\bm S\bm A B} ) \leq \mathcal{N}_B \big(E_{\rm max}(\beta_B)+E_\psi - E_{\min}^{(m)}\big), 
\end{align}
where $ \mathcal{N}_B \big(E\big)$ is, by definition, the number of battery levels with energy no greater than $E$. Recall that $H_{\bm S\bm A}$ is a sum of $2m$ operators $H_S$ acting on different $S_i$ and $2m$ operators $H_A=-H_S$ acting on different $A_i$, leading to $E_{\min}^{(m)} = - 2 m (E_{\max}-E_{\min})$, where $E_{\min}$ ($E_{\max}$) is the minimal (maximal) eigenvalue of $H_S$. Introducing the pseudo-norm $\|H_S\|_\Delta := E_{\max}-E_{\min}$ we obtain for pure $\beta_{B}$
\begin{align} \label{eq: rank 3 temp}
    {\rm Rank}({\tnu}^{(m)}_{\bm S\bm A}) \leq \mathcal{N}_B \big(E_{\rm max}(\beta_B)+E_\psi + 2m \|H_S\|_\Delta\big).
\end{align}
 Finally, if $\beta_B$ is not pure it is a mixture of ${\rm Rank}(\beta_B)$ terms, implying
\begin{align} 
     {\rm Rank}({\tnu}^{(m)}_{\bm S\bm A}) &\leq  {\rm Rank}(\beta_B) \,\mathcal{N}_B \big(E_{\rm max}(\beta_B) + 4m \|H_S\|_\Delta\big) \\
     & \leq \mathcal{N}_B \big(E_{\rm max}(\beta_B)\big) \times \mathcal{N}_B \big( E_{\rm max}(\beta_B)+ 4m \|H_S\|_\Delta\big).
\end{align} 

Now consider a family of batteries $\{ H_B^{(\epsilon)},\beta_B^{(\epsilon)}\}_\epsilon$ suitable for the target gate. As we will see later, in the limit of $\epsilon\to 0$ we will always choose $m$ such that $ 4m \|H_S\|_\Delta\leq  E_{\rm max}(\beta_B^{(\epsilon)})$. In this case we get 
\begin{align} \label{eq: bound rank prop}
     {\rm Rank}({\tnu}^{(m)}_{\bm S\bm A})  \leq \mathcal{N}_B \big(E_{\rm max}(\beta_B^{(\epsilon)})\big) \times \mathcal{N}_B \big( 2 E_{\rm max}(\beta_B^{(\epsilon)})\big) \leq {\rm poly}(\epsilon^{-1}),
\end{align} 
for any \textit{proportionate} family of batteries. This is the inequality that we will use to prove (ii) of the main result.

\subsubsection{The final bounds}

We can now combine all of the above bounds to obtain the following result

\begin{rslt}[Regularity of the entropic coherence] \label{res: resource reg}
    In the limit $m\to\infty$ the difference in the relative entropy of coherence between the ideal ${\nu}^{(m)}_{\bm S\bm A}$ and the approximate ${\tnu}^{(m)}_{\bm S\bm A}$ states is upper bounded by
    \begin{align}  \label{Prop:RegularityRecDef}
        C&({\nu}^{(m)}_{\bm S\bm A},H_{\bm S\bm A}) - C({\tnu}^{(m)}_{\bm S\bm A},H_{\bm S \bm A})  
\leq   4 m \sqrt{\epsilon} \Big(\frac{d_S^2}{2}   \log(2m) + 4 m\log(d_S)\Big) + 2\, h_2\big(4m \sqrt{\epsilon}\big).
    \end{align}
    For a \textit{proportionate} battery the following bound also holds
    \begin{align}  \label{Prop:RegularityRecDef2}
        C&({\nu}^{(m)}_{\bm S\bm A},H_{\bm S\bm A}) - C({\tnu}^{(m)}_{\bm S\bm A},H_{\bm S \bm A})  
\leq   4 m \sqrt{\epsilon} \Big(\frac{d_S^2}{2}   \log(2m) +\log\left({\rm poly}(\epsilon^{-1})\right)\Big) + 2\, h_2\big(4m \sqrt{\epsilon}\big),
    \end{align}
under the assumption $ 4m \|H_S\|_\Delta\leq  E_{\rm max}(\beta_B^{(\epsilon)})$.
\end{rslt}
\begin{proof}
    The bounds are an immediate consequence of combining 
    Eqs.~(\ref{eq: taming bound}, 
    \ref{app eq: bound rank gen}, \ref{eq:distance4m}) with Eqs.~\eqref{eq: bound rank no ass} or \eqref{eq: bound rank prop} depending on whether the battery is \textit{proportionate}. Finally, to simplify Eq.~\eqref{app eq: bound rank gen} we also used that for large enough $m$
    \begin{equation}
    \log \left[2^{d_{S}(d_{S}-1)}\binom{2m + \frac{d_{S}(d_{S}-1)}{2}}{2m} \right]\leq \log((2m)^{\frac{d_S^2}{2}}) = \frac{d_S^2}{2}\log(2m).
    \end{equation}
\end{proof}
\subsection{Step 3: Choosing the best initial state.}

As discussed in the proof outline, to obtain a lower bound on $C(\beta,H_{B})$ it now remains to combine the bounds on resource production (Results \ref{lem:Entropy-twirled} and \ref{res: res prod qubit}) and resource regularity (Result \ref{res: resource reg}) derived in the last two sections, and choose the optimal scaling of $m =  m( \epsilon)$ to take the limit $\epsilon\to 0$. To do so we consider separately the case of a general or \textit{proportionate} battery.

\subsubsection{Main result (i) i.e. any battery}

For a generic battery, Eq.~\eqref{eq: LB-UB} and the previous results give the following bound on the entropic coherence of the battery
\begin{equation}
     \label{Eq:Last-general-eq}
  C(\beta,H_{B}) \geq  \frac{r_{2}(V_{S},H_S)}{2} \log(2 \pi \; e \; m \lambda_2(V_{S},H_S)) -16 \log(d_S) \sqrt{\epsilon} \, m^2- 2 d_S^2\sqrt{\epsilon}  \,  m   \log(2m) - 2\, h_2\big(4m \sqrt{\epsilon}\big) - o(1).
\end{equation}
To obtain the strictest inequality we set the following scaling of the number of copies 
\begin{equation} \label{Eq:Best-m}
    m(\epsilon):=\left \lfloor \frac{\sqrt{r_{2}(V_{S},H_S)}}{8\sqrt{\ln(2)\log(d_{S})}\epsilon^{1/4}} \right \rfloor,
\end{equation}
with $ \sqrt{\epsilon}  \,  m   \log(2m), h_2\big(4m \sqrt{\epsilon}\big)  \to 0$. This choice leads to  
\begin{align}
 C({\beta}, {H}_B) &\geq \frac{r_{2}(V_{S},H_S)}{8} \log\left( \frac{(2 \pi e\;\lambda_2(V_{S},H_S)\sqrt{r_{2}(V_{S},H_S)} )^{4}}{8^4 \ln^2(2)\log^2(d_{S})\epsilon}\right) - \frac{r_{2}(V_{S},H_S)}{4\ln(2)} - o(1), \\ 
&=\frac{r_{2}(V_{S},H_S)}{8}  \log\left(\frac{\sigma(V_S,H_S)}{ \epsilon}\right)-o(1)
\end{align}
where we introduced
\begin{equation} \label{def:sigma(V_{S})}
    \sigma(V_S,H_S):=\frac{\pi^{4} e^{2}}{256\ln^2(2)} \, \frac{ (\lambda_2(V_S,H_S))^4 ( r_2(V_S,H_S))^{2}}{\log^2(d_S)}.
\end{equation}

\subsubsection{Main result (ii), \textit{proportionate} battery}

For a \textit{proportionate} battery combining \eqref{eq: LB-UB} with the previous results gives the following bound on the entropic coherence of the battery
\begin{equation}
     \label{Eq:Last-prop-eq}
  C(\beta,H_{B}) \geq  \frac{r_{2}(V_{S},H_S)}{2} \log(2 \pi \; e \; m \lambda_2(V_{S},H_S)) - 2 d_S^2 \sqrt{\epsilon}  \,  m   \log(2m) - 4 m \sqrt{\epsilon}\,  \log\left({\rm poly}(\epsilon^{-1})\right)  - 2\, h_2\big(4m \sqrt{\epsilon}\big) - o(1).
\end{equation}
Here, we will choose the number of copies as
\begin{equation}
   m(\epsilon):=\left \lfloor\frac{\gamma}{\sqrt{\epsilon}\log{\epsilon^{-1}}}\right \rfloor,
\end{equation}
where $\gamma>0$ is a constant that will be determined later. First, we need to verify that this choice satisfies the assumption $ 4 m(\epsilon) \|H_S\|_\Delta\leq  E_{\rm max}(\beta_B^{(\epsilon)})$. This follows from Theorem 1 of \cite{Chiribella}, which implies that the energy of a suitable battery must satisfy 
\begin{equation}
E_{\rm max}(\beta_B^{(\epsilon)})
\geq 
\frac{C}{\sqrt{\epsilon}}
\geq 
\frac{4 \gamma \|H_S\|_\Delta}{\sqrt{\epsilon}\log(\epsilon^{-1})}
\end{equation}
in the limit of small $\epsilon$.

In Eq.~\eqref{Eq:Last-prop-eq} we find that $h_2\big(4m \sqrt{\epsilon}\big)\to 0$,
\begin{equation}
    4  \frac{d_S^2}{2} \sqrt{\epsilon}  \,  m   \log(2m) = \frac{2d_S^2\gamma}{\log(\epsilon^{-1})}\log\left(\frac{2\gamma}{\sqrt{\epsilon}\log(\epsilon^{-1})}\right)=d_S^2\gamma+o(1),
\end{equation}
and
\begin{align}
4 m \sqrt{\epsilon}\,  \log\left({\rm poly}(\epsilon^{-1})\right) = \frac{4\gamma \log\left({\rm poly}(\epsilon^{-1})\right) }{\log(\epsilon^{-1})} \leq 4 \gamma \,  \alpha + o(1)
\end{align}
where $\alpha>0$ is the dominant power of the polynomial, i.e. such that ${\rm poly}(\epsilon^{-1}) \leq \epsilon^{-\alpha}$ in the limit $\epsilon\to 0$.
Combining with the positive term we obtain
\begin{align}
C({\beta}, {H}_B) &\geq \frac{r_{2}(V_{S},H_S)}{4} \log\left( \frac{(2 \pi  e  \lambda_2(V_{S},H_S))^2}{\epsilon}\right)  -\frac{r_2(V_S,H_S)}{2} \log ( \log\epsilon^{-1})+ \frac{r_{2}(V_{S},H_S)}{2}\log(\gamma)-  \gamma (d_S^2 +4 \alpha) -o(1).
\end{align}
The rhs is maximized for the choice $\gamma = \frac{r_2(V_S,H_S)}{4\ln(2)( \frac{d_S^2}{2}+2\alpha)}$, leading to the final bound
\begin{align}
    C({\beta}, {H}_B) &\geq \frac{r_{2}(V_{S},H_S)}{4} \log\left( \frac{(2 \pi  e  \lambda_2(V_{S},H_S))^2}{\epsilon}\right)  -\frac{r_2(V_S,H_S)}{2} \log ( \log\epsilon^{-1})+ \frac{r_2(V_S,H_S)}{2}\log\left(\frac{r_2(V_S,H_S)}{ e\ln(2) (2d_S^2+8\alpha)}\right) -o(1)\nonumber \\
    & = \frac{r_{2}(V_{S},H_S)}{4} \log\left(\frac{\sigma'(V_S,H_S)}{\epsilon(\log \epsilon^{-1})^2}\right)-o(1),
\end{align}
where we introduced
\begin{equation} \label{def:sigma'(V_{S})}
    \sigma'(V_S,H_S) :=  \frac{\pi^2}{4\ln^2(2)} \left(\frac{\lambda_2(V_S,H_S)\,r_2(V_S,H_S)}{\frac{d_S^2}{2}+2\alpha}\right)^2.
\end{equation}

\section{Proving corollaries \ref{Cor:Energy} and \ref{Cor:QFI}}
\label{app: corrs}

We begin by proving that at fixed energy distribution, pure states have higher entropic coherence: 
\begin{lem} \label{lem:Pure>mixed}
For any Hamiltonian $H$ and state $\sigma$, let $\ket{\sigma}:=\underset{\ket{E} \in \mathcal{B}[H]}{\sum}\sqrt{\bra{E}\sigma\ket{E}}\ket{E}$ where $\mathcal{B}[H]\subseteq {\tt Eig}[H]$ is an orthonormal basis, then
    \begin{equation}
        C(\ket{\sigma},H) \geq C(\sigma,H),
    \end{equation}

\end{lem}
\begin{proof}
Let $\Pi_{E}$ be the projectors on the fixed energy subspaces of $H$, and $\sigma_{E}:=\frac{\Pi_{E}\sigma\Pi_{E}}{P_{E}}$ where $P_{E}:=\tr[\sigma \Pi_{E}]$; then we have 
\begin{align}
& C(\ket{\sigma},H) - C(\sigma,H)= S(\sigma)+S(\mathcal{G}_{H}(\ket{\sigma}))-S(\mathcal{G}_{H}(\sigma))\\
&=S(\sigma)+S(\{ P_{E}\}_{E})-S\left(\underset{E \in \mathcal{E}[H]}{\bigoplus}P_{E}\sigma_{E}\right) \geq 0
\end{align}
where the inequality is a consequence of the strong subadditivity of entropy \cite{StrongSubadditivity2002}.Indeed,
Consider the system $S$ in the state $\sigma_S$, and its purification $\Psi_{SA}$ on $S$ composed with an auxiliary system $A$. Consider the isometry $W_{SR} = \suma{E} \Pi^{(E)}_S \otimes \ketbra{E}{0}_R$ (where $\braket{E}{E'}:=\delta_{E,E'}$), applied on $S$ and a register $R$, initially in the state $\ketbra{0}_R$. The isometry can be realized by a joint unitary on $SR$, hence the final state $\rho_{SR} = W_{SR} \sigma_S \otimes \ketbra{0}_R W_{SR}^\dag$ has the same entropy as the initial state of the system $S(\rho_{SR})=S(\sigma_S\otimes \ketbra{0}_R)=S(\sigma_S)$. Denoting the final pure state of the three systems $\Phi_{SAR}=W_{SR} \Psi_{SA} \otimes \ketbra{0}_R W_{SR}^{\dagger}$ we find for this state and its marginals
\begin{align}
    S(SAR) &= S( \Phi_{SAR}) =0, \\
    S(SA) &= S(R) = S\left( \sum_E P_E \ketbra{E}_R \right),\\
    S(SR) &= S(\rho_{SR}) = S(\sigma_{S}),\\
    S(S)  &= S\left(\sum_E \Pi_S^{(E)} \rho_S \Pi_S^{(E)}\right) = S\left(\underset{E \in \mathcal{E}[H]}{\bigoplus}P_{E}\sigma_{E}\right).
\end{align}
Combining these identities with the strong subadditivity of the entropy~\cite{StrongSubadditivity2002} implies 
\begin{equation}
0\leq S(SA)+ S(SR)-S(S) -S(SAR) = S(\{P_{E}\}_{E}) + S(\sigma_S)-S\left(\underset{E \in \mathcal{E}[H]}{\bigoplus}P_{E}\sigma_{E}\right),
\end{equation}
proving the statement.

We now observe that $\ket{\sigma}, \sigma$ have the same energy distribution, meaning that $\tr [\sigma H]=\bra{\sigma}H\ket{\sigma},{\rm Var}(\sigma,H)={\rm Var}(\ket{\sigma},H)$. Thus the states optimizing Eq.~(\ref{Eq:minenergy}) can always be chosen pure, and the only variable of optimization is the probability distribution $\{P_{E}\}_{E}$. We have thus transformed the optimizations into their classical counterparts, finding the distribution with least average energy, respectively, variance, at fixed minimum entropy. The solution to these problems is well known \cite{MaxEntropyCoverTomas,MaxEntropySecondMomentFixCont, RioulMassey}, and is summarized by Eqs.~(\ref{Eq:minenergy}) and~(\ref{eq:variance-from-coherence}).

Consider now the set of Hamiltonians with a linear number of energy levels, i.e. $\mathcal{L}_{\eta}:= \{H: \mathcal{N}(H,E)\leq 1+\eta E\}$. We prove that among Hamiltonians $H\in \mathcal{L}_{\eta}$, the one with the least energetic states at fixed entropy of coherence is {the harmonic oscillator} $H_{h.o.}^{(\eta)}=\underset{n=0}{\overset{\infty}{\sum}} \eta^{-1} \,  n \, \ketbra{n}{n}$ of frequency $\eta^{-1}$. To see this, consider a generic $H$. If there exist $\delta>0,\ket{E'}\in {\tt Eig}[H]$ such that $H':=H-\delta \ketbra{E'}{E'}\in \mathcal{L}_{\eta}$ and $(E'-\delta) \notin \mathcal{E}[H]$, then for all pure states
 $\ket{\sigma}$ we have
 \begin{align}
     \bra{\sigma}H'\ket{\sigma}\leq \bra{\sigma}H\ket{\sigma}; \quad C(\ket{\sigma},H')\geq C(\ket{\sigma},H),
 \end{align}
 where $C(\ket{\sigma},H')> C(\ket{\sigma},H)$ can occur when $E'$ is a degenerate eigenvalue and $\bra{E'}\sigma\ket{E'}< \tr[\Pi_{E'}\sigma]$ and the shift lifts the degeneracy. The proof is complete since $H_{h.o.}^{(\eta)}$ is the unique Hamiltonian (up to irrelevant changes of basis) for which such $\delta,E'$ never exist. Thus for every $H\in \mathcal{L}_{\eta}$
\begin{equation}
    \underset{\rho: C(\rho,H)\geq C}{\min} \tr  [H\rho] \geq  \frac{\tr [H_{h.o.}^{(\eta)} e^{-\gamma_{C} H_{h.o.}^{(\eta)}}]}{\tr [e^{-\gamma_{C} H_{h.o.}^{(\eta)}}]}.
\end{equation}
Finally, by direct computation $\gamma_{C}=e\eta\,2^{-C}+O(\eta\,2^{-2C})$ and $\frac{\tr [H_{h.o.}^{(\eta)} e^{-\gamma_{C} H_{h.o.}^{(\eta)}}]}{\tr [e^{-\gamma_{C} H_{h.o.}^{(\eta)}}]}=\eta^{-1} \, \left(\frac{2^{C}}{e}-\frac{1}{2}+O(2^{-C})\right)$.

Corollary~\ref{Cor:Energy} then follows immediately from (i) of the main result.
\end{proof}
To prove Corollary ~\ref{Cor:QFI} we first prove an analogus statement for the variance of the battery state.
\begin{lem}[Variance {constraints} to implement NEPGs]\label{lem:var}
Let $H_{B}$ be the Hamiltonian of a harmonic oscillator of frequency $\omega$ and $\{H_B,\beta_B^{(\epsilon)}\}_{\epsilon}$ be suitable for the NEPG $\mathcal{V}_S $, then 
\begin{equation}
  {\rm Var}(\beta^{(\epsilon)}_{B},H_{B})\geq  \left(\frac{\omega^{2}\sigma(V_{S},H_{S})^{\frac{r_{2}(V_{S},H_{S})}{4}}}{2e\pi}-o(1)\right)\epsilon^{-\frac{r_{2}(V_{S},H_{S})}{4}}.
\end{equation}
\end{lem}

\begin{proof}
 According to (i) of the Main Result, it is necessary for the battery system to satisfy 
\begin{equation}\label{eq:again}
    C(\beta^{(\epsilon)}_{B},H_{B})\geq  \frac{r_{2}(V_{S})}{8}\log(\frac{\sigma(V_{S})}{\epsilon})-o(1).
\end{equation}
 We now prove that for an harmonic oscillator of frequency $\omega$ this inequality alone implies the lower bound of the thesis, to do so is sufficient to prove 
\begin{equation} \label{eq:variance-from-coherence}
    {\rm Var}(\beta_B,H_B)
    \geq
    \omega^2
    \left(
        \frac{2^{2C(\beta_B,H_B)}}{2\pi e}
        -\frac{1}{12}
    \right).
\end{equation}
and plug Eq.\eqref{eq:again} in the latter.
 
 Because of Lemma \ref{lem:Pure>mixed} the minimum-variance state at fixed entropic coherence can be assumed to be pure.
 For the harmonic oscillator
\begin{equation}
    H_B
    =
    \omega\sum_{n=0}^{\infty}n\ketbra{n}{n},
\end{equation}
the problem is therefore reduced to its classical counterpart:
finding the probability distribution of a random variable $X$ taking
values in $\mathbb{N}$, with probabilities
$p_n:={\rm Pr}(X=n)$, that minimizes the variance at fixed entropy
\cite{MaxEntropyCoverTomas,MaxEntropySecondMomentFixCont}.

Indeed, for the pure state $ \ket{\beta} =\sum_{n=0}^{\infty}\sqrt{p_n}\,e^{i\phi_n}\ket{n}$ we have
\begin{equation}
    C(\ket{\beta},H_B)=H(X),
    \qquad
    {\rm Var}(\ket{\beta},H_B)
    =
    \omega^2{\rm Var}(X).
\end{equation}
We now use the entropy--variance inequality for integer-valued random
variables \cite{RioulMassey},
\begin{equation}
    H(X)
    \leq
    \frac{1}{2}
    \log\left[
        2\pi e
        \left(
            {\rm Var}(X)+\frac{1}{12}
        \right)
    \right].
    \label{eq:Massey-entropy-variance}
\end{equation}
Applying it to a distribution with $H(X)=C$ and inverting the
inequality gives
\begin{equation}
    {\rm Var}(X)
    \geq
    \frac{2^{2C}}{2\pi e}
    -\frac{1}{12}.
\end{equation}
Finally, using
\begin{equation}
    {\rm Var}(\beta_B,H_B)
    =
    \omega^2{\rm Var}(X)
\end{equation}
and setting $C=C(\beta_B,H_B)$ proves
Eq.~\eqref{eq:variance-from-coherence}.
\end{proof}

We can now finally prove Corollary \ref{Cor:QFI}, which we restate for the reader convenience. 
\begin{cor*}[QFI constraints to implement NEPGs]
Let $H_{B}$ be a harmonic oscillator of frequency $\omega$ and $\{H_B,\beta_B^{(\epsilon)}\}_{\epsilon}$ be suitable for the NEPG $\mathcal{V}_S $, then 
\begin{align} \label{eq: QFI bound corr}
  {\rm QFI}(\beta,H_{B})\geq &\left(\frac{2\omega^{2}\sigma(V_{S},H_{S})^{\frac{r_{2}(V_{S},H_{S})}{4}}}{e\pi}-o(1)\right)(\epsilon d_{S})^{-\frac{r_{2}(V_{S},H_{S})}{4}}.
\end{align}
\end{cor*}

\begin{proof}

We exploit the fact that the QFI is the convex roof of the variance~\cite{QFI-ConvexRoof} to transform lemma \eqref{lem:var} {into a bound on} the QFI. Consider the Choi-infidelity between channels, i.e. $\epsilon_{C}(\mathcal{V},\mathcal{W}):=1-F(\sigma^{C}_{\mathcal{V}},\sigma^{C}_{\mathcal{W}})$, where $\sigma^{C}_{\mathcal{V}},\sigma^{C}_{\mathcal{W}}$ are respectively the Choi states of the channels $\mathcal{V},\mathcal{W}$ \cite{Choi1975}. Let $\epsilon_{C}(\epsilon):=\epsilon_{C}(\tcV^{(\epsilon)},\cV)$ be the Choi-infidelity corresponding to a worst-case infidelity $\epsilon$. As proved in \cite{UDpaper}, $\epsilon_{C}$ and $\epsilon_{wc}$ are equivalent up to the dimension of the system, $\epsilon_{C}(\epsilon)d_{S} \geq \epsilon \geq \epsilon_{C}(\epsilon)$. Then we have 
\begin{equation}
  {\rm Var}(\beta_{B}^{(\epsilon)},H_{B})\geq \left(\xi-o(1)\right)(d_{S}\epsilon_{C}(\epsilon))^{-r},
\end{equation}
where we made the substitution $r:=\frac{r_{2}(V_{S},H_{S})}{4},\xi=\frac{\omega^{2}\sigma(V_{S},H_{S})^{r}}{2e\pi}$.
We now recall that QFI
  equals 4 times the convex roof of the variance \cite{QFI-ConvexRoof}, and thus there exists an ensemble representation of 
 $\beta^{(\epsilon)}$, which we call $\{ p_{j}^{(\epsilon)},\ket{j^{(\epsilon)}}\}_{j}$, 

such that
\begin{equation} \label{QFIConvexRoofExpression}
  {\rm QFI}(\beta^{(\epsilon)}_{B},H_B) = 4  \sum_{j} p_j^{(\epsilon)} \, \mathrm{Var}[\ket{j^{(\epsilon)}}, H_{B}]
\end{equation}
For each ensemble component, define
\begin{equation}
    \tcV_{j^{(\epsilon)}}(\cdot)
    :=
    \tr_B\!\left[
    \mathcal{U}_{SB}
    \left(
    \cdot\otimes
    \ketbra{j^{(\epsilon)}}{j^{(\epsilon)}}
    \right)
    \right]
\end{equation}
and
\begin{equation}
    \epsilon_C^{(j)}(\epsilon)
    :=
    \epsilon_C\!\left(
    \tcV_{j^{(\epsilon)}},
    \mathcal{V}_S
    \right).
\end{equation}
Since the implemented channel is linear in the battery state,
\begin{equation}
    \tcV^{(\epsilon)}
    =
    \sum_jp_j^{(\epsilon)}
    \tcV_{j^{(\epsilon)}}.
\end{equation}
Moreover, since the Choi state of $\mathcal{V}_S$ is pure, the Choi
infidelity is affine in the implemented channel, and hence
\begin{equation}
    \epsilon_C(\epsilon)
    =
    \sum_jp_j^{(\epsilon)}
    \epsilon_C^{(j)}(\epsilon).
    \label{eq:average-Choi-error}
\end{equation}

Set
\begin{equation}
    r:=\frac{r_2(V_S,H_S)}{4},
    \qquad
    \xi:=
    \frac{\omega^2\sigma(V_S,H_S)^r}{2e\pi}.
\end{equation}
The state-independent entropy--variance bound gives
\begin{align}
    {\rm Var}\!\left(\ket{j^{(\epsilon)}},H_B\right)
    &\geq
    \omega^2\left(
    \frac{
    2^{2C(\ket{j^{(\epsilon)}},H_B)}
    }{2e\pi}
    -\frac{1}{12}
    \right)
    \nonumber\\
    &\geq
    \xi\,2^{-2o(1)}
    \left(d_S\epsilon_C^{(j)}(\epsilon)\right)^{-r}
    -\frac{\omega^2}{12}.
\end{align}
Consequently,
\begin{equation}
    {\rm Var}\!\left(\ket{j^{(\epsilon)}},H_B\right)
    \geq
    \left[
        \xi\,2^{-2o(1)}
        -\frac{\omega^2}{12}
        \left(d_S\epsilon_C^{(j)}(\epsilon)\right)^r
    \right]
    \left(d_S\epsilon_C^{(j)}(\epsilon)\right)^{-r},
    \label{eq:uniform-pure-variance-bound}
\end{equation}
where the $o(1)$ is uniform and vanishes as
$d_S\epsilon_C^{(j)}(\epsilon)\rightarrow0$.

Define
\begin{equation}
    x_j:=d_S\epsilon_C^{(j)}(\epsilon),
    \qquad
    x:=d_S\epsilon_C(\epsilon)
    =
    \sum_jp_j^{(\epsilon)}x_j,
\end{equation}
and consider the set
\begin{equation}
    G_\epsilon:=\left\{j:x_j\leq\sqrt{x}\right\},
    \qquad
    P_\epsilon:=\sum_{j\in G_\epsilon}p_j^{(\epsilon)}.
\end{equation}
Since $x_j>\sqrt{x}$ for every $j\notin G_\epsilon$, we have
\begin{equation}
    x
    \geq
    \sum_{j\notin G_\epsilon}p_j^{(\epsilon)}x_j
    \geq
    \sqrt{x}\left(1-P_\epsilon\right),
\end{equation}
and therefore
\begin{equation}
    P_\epsilon\geq1-\sqrt{x}.
\end{equation}

Using Eq.~\eqref{eq:uniform-pure-variance-bound}, positivity of the
variance, and the convexity of $y\mapsto y^{-r}$, we obtain
\begin{align}
    {\rm QFI}\!\left(\beta_B^{(\epsilon)},H_B\right)
    &\geq
    4\left[\xi\,2^{-2o(1)}-\frac{\omega^2}{12}x^{r/2}\right]
    \sum_{j\in G_\epsilon}p_j^{(\epsilon)}x_j^{-r}
    \nonumber\\
    &\geq
    4\left[\xi\,2^{-2o(1)}-\frac{\omega^2}{12}x^{r/2}\right]
    P_\epsilon^{\,r+1}
    \left(
    \sum_{j\in G_\epsilon}p_j^{(\epsilon)}x_j
    \right)^{-r}
    \nonumber\\
    &\geq
    4\left[\xi\,2^{-2o(1)}-\frac{\omega^2}{12}x^{r/2}\right]
    (1-\sqrt{x})^{r+1}x^{-r}
    \nonumber\\
    &=
    \left(4\xi-o(1)\right)
    \left(d_S\epsilon_C(\epsilon)\right)^{-r}.
\end{align}
Finally, $\epsilon_C(\epsilon)\leq\epsilon$, and hence
\begin{equation}
    {\rm QFI}\!\left(\beta_B^{(\epsilon)},H_B\right)
    \geq
    \left(4\xi-o(1)\right)
    \left(d_S\epsilon\right)^{-r}.
\end{equation}
Substituting the values of $r$ and $\xi$ gives
\begin{equation}
    {\rm QFI}\!\left(\beta_B^{(\epsilon)},H_B\right)
    \geq
    \left[
    \frac{
    2\omega^2
    \sigma(V_S,H_S)^{\frac{r_2(V_S,H_S)}{4}}
    }{e\pi}
    -o(1)
    \right]
    \left(d_S\epsilon\right)^{
    -\frac{r_2(V_S,H_S)}{4}
    },
\end{equation}
which proves the corollary.
\end{proof}

\section{Entropy of sums of discrete i.i.d. random variables }

\label{app: iid}

In this appendix we give a brief review of the bound on the entropy of sums of i.i.d. discrete random variables, used in Result~\ref{lem:Entropy-twirled}. The goal of the appendix is to present the minimal definitions needed to apply the bound and to rely entirely on the results published in~\cite{OurIIDEntropy}.\\

\renewcommand{\vec}[1]{{\bm{#1}}}

Let $\{X_{i}\}_{i \in \mathbb{N}}$ be a sequence of i.i.d. random variables that take values in a finite set of real numbers $\chi:=\{x_{1},..x_{d} \}$. We also use the notation $\vec{x}:=(x_{1},x_{2},..x_{d})$ and $\vec{p}:=(p_{1},...p_{d})$ for the corresponding vectors of probabilities. We are interested in lower-bounding the entropy of their sum 
\begin{equation}
T_{N}:= \sum_{i=1}^{N}X_{i}. 
\end{equation}

First, we briefly discuss the case of the so-called lattice random variables. Note in particular that deterministic and binary random variables ($d=1,2$) are always lattice, as follows from the following definition. 
\begin{defn}\label{defn:lattice}
$X$ is a {\bf lattice random variable}, if it takes values in a {\bf lattice set} $\chi$, i.e. such that 
\begin{align}\label{app: lattice}
    \chi \subset \big \{ h \, n + a : n \in \mathbb{Z}\big\}
\end{align}
for  some $h, a \in \mathbb{R}$. The {\bf maximal span} $h(X)=h(\chi)$ is the maximal real number $h$ for which Eq.~\eqref{app: lattice} can be fulfilled. 
\end{defn}
For a lattice random variable, an asymptotically tight expression for the entropy is known.
\begin{thm}\cite{EntropySumIID}\label{Th:LatticeEntropy}
    Let $X$ be a lattice random variable with maximal span $h(\chi)>0$. Then in the large $N$ limit
    \begin{equation}\label{eq: greek thrm}
        S(T_{N})= \frac{1}{2}\log(N)+\frac{1}{2}\log(2\pi e \frac{{\rm Var}(X)}{h(\chi)^{2}})+o(1).
    \end{equation}
\end{thm}

Next, we consider general discrete random variables. We first introduce some notion of the incommensurability of the set of values $\chi$ it may take. 
For this, consider a partition of a set $\chi$ into $k$ disjoint subsets $\{\chi_{j} \}_{j=1..k}$, and let 
\begin{equation}
\mathcal{S}_{m}(\chi_j) := \left\{ \sum_{x_{i}\in \chi_{j}}n_{i} x_{i} : 
n_i \in \mathbb{N}, \sum_i n_i = m
\right \}
\end{equation}
be the set of numbers that can be obtained by summing $m$, not necessarily different, numbers from $\chi_{j}$. For convenience we also introduce the notation 
$ \mathbb{N}^{k}_{N}:=\left\{ \bm{n}\in \mathbb{N}^{k} : \sumab{i=1}{k} n_{i}=N \right\}$,
for the set of natural $k$-vectors $\bm n$ whose elements sum up to $N$ (with $0 \in \mathbb{N}$).  With its help, we define the following property.

\begin{defn}\label{def: incom}
A collection of disjoint sets $\{\chi_{j} \}_{j=1..k}$ (e.g. a partition of $\chi$) is called \textbf{incommensurable} if all $\chi_j$ are non-empty, and for all $N\in \mathbb{N}$, $\bm{m},\bm{n}\in \mathbb{N}^{k}_{N}$ and $ \vec{y}= (y_1,\dots,y_k)$, $\vec{y}'= (y_1',\dots,y_k')$ with $y_j\in \mathcal{S}_{m_{j}}(\chi_{j})$ and $y_j' \in \mathcal{S}_{n_{j}}(\chi_j)$ we have 
\begin{equation}
\sum_{j=1}^k y_j = \sum_{j=1}^k y_j' \implies y_j =y_j' \quad \forall j \in \{1, ...k\} .
\end{equation}
\end{defn}
The idea behind this definition is that for an \textit{incommensurable} partition, the observed value $t=T_N$ \textit{uniquely} specifies the totals $(y_1,\dots,y_k)$ obtained in each subset. Hence, there is a bijection between the random variable $T_{N}$ and the $k$ random variables $(Y_{1}^{(N)},\dots, Y_{k}^{(N)})$
\begin{equation}
Y_{j}^{(N)}:=\sum_{i=1}^{N} X_{i} \, \mathbf{1}_{\chi_{j}}(X_{i}), \qquad \text{where} \qquad 
\mathbf{1}_{\chi_{j}}(x) = 
\begin{cases}
1 & x \in  \chi_{j} \\
0 & \text{otherwise}
\end{cases}
\end{equation}
are indicator functions verifying if $X$ takes value in $\chi_j$, and each $Y_{j}^{(N)}$ only sums the values drawn from $\chi_j$. \\

The cases where $\chi$ can be exactly partitioned into incommensurable sets are quite specific, more generally one may consider the following definitions.

\begin{defn}\label{def: k can}
    Consider a discrete set of  real numbers $\chi$. A collection of disjoint sets $\{\chi_{1},..\chi_{k}\}$ is called a \textbf{k-canonical prepartition} (or canonical prepartition) of $\chi$ if it satisfies the following properties: 
\begin{enumerate}
\item $\chi_j \subset\chi$ 
 \item $\chi_{j}$ is a non-empty lattice for $j=1,..,k$ (definition \ref{defn:lattice}).
 \item $\{\chi_{j}\}_{j=1,..,k} $ is incommensurable (definition \ref{def: incom}).
\end{enumerate}
If $|\chi_j| =1$ for all $j=1,\dots,k$ we say that the canonical prepartition is {\bf degenerate}. 
\end{defn}

One notes that if the sets $\chi_{1},..\chi_{k}$ all have a single element (i.e. the prepartition is degenerate), one can merge any two of them to obtain a $(k-1)$-canonical prepartition (non-degenerate). This observation, motivates the following definition.

\begin{defn} \label{Def:r-Set}
    Let $\chi:=\{x_{1},..x_{d}\}$  be a discrete set of real numbers with $d\geq 2$.  We define its {\bf incommensurability rank} $r(\chi)$ as the maximal integer $1\leq r\leq d-1$ such that $\chi$ admits a non-degenerate $r$-canonical prepartition. For a random variable $X$ taking values in $\chi$ we also call $r(X)=r(\chi)$ its incommensurability rank.
\end{defn}
One can also show that the incommensurability rank of a set is invariant under shifts, and $r(X)=1$ if and only if $X$ is a nondeterministic lattice random variable. For a general random variable determining $r(X)$ is not necessarily straightforward.  However the following result provides a useful lower bound on $r(X)$.

\begin{lem}\label{lem: part span}
 Any finite set $\chi=\{x_1,\dots,x_d\}$ admits a $Q$-\textit{canonical partition}, where $Q:={\rm dim}_{\mathbb{Q}} [{\rm span}_{\mathbb{Q}}(\chi)]$ is the dimension of the rational span of $\chi$ seen as  vector space over the field of rational numbers $\mathbb{Q}$. That is,
\begin{equation}
{\rm span}_{\mathbb{Q}}(\chi):= \left \{ \underset{x \in \chi}{\sum} q_{x} x : q_{x} \in \mathbb{Q}\right\}
\end{equation}
 is the set containing all linear combinations of elements of $\chi$ with {\it rational} coefficients. It follows then that a random variable $X$ taking values in $\chi$ satisfies $r(X)\geq Q-1$.
\end{lem}

The definitions of $k$-canonical prepartitions and incommensurability rank give rise to the following result and its immediate corollary. The latter was used in the proof of Result~\ref{lem:Entropy-twirled}.

\begin{thm}
 \label{Th:General-iid-entropy} \cite{OurIIDEntropy}
 Let $X_1,\dots, X_N$ be a collection of i.i.d. random variables taking values in $\chi$, a set of real numbers that admits a $k$-canonical prepartition $\{\chi_{1},..\chi_{k}\}$ (Definition \ref{def: k can}) with $p_j = {\rm Pr}(X_i \in \chi_j)$ for $j= 1,\dots,k$, and $q =\sumab{j=1}{k} p_j$. In addition, let $\chi_{1},..\chi_{s\leq k}$ be the only subsets with one element; then the following bounds hold
 \begin{align}\label{eq: lambda1}
      s<k: \quad    S\left(\sum_{i=1}^N X_i \right) &\geq 
          \frac{k}{2} \log(2\pi e  N \lambda_1)-o(1) \qquad \quad \,\, \, \, \lambda_{1}:=\left( p_1\dots p_s \left(1- \sum_{j=1}^s \frac{p_j}{q}\right)\prod_{j=s+1}^k p_j \frac{{\rm Var}(\tilde X^{(j)})}{h(\chi_j)^{2}}\right )^{1/k} \\
    s = k: \quad   S\left(\sum_{i=1}^N X_i \right) &\geq  \frac{k-1}{2} \log(2\pi e  N \lambda_{2})-o(1)  \qquad \lambda_{2}:= \left (\frac{p_{1}\dots p_{k} }{q^{k}}\right )^{\frac{1}{k-1}}
 \end{align}
where we used the notation $\tilde X^{(j)}\in \chi_j$ with ${\rm Pr}(\tilde X^{(j)}=x) ={\rm Pr}(X=x| X \in \chi_j)$ and $h(\chi_j)$ is the maximal span of the lattice $\chi_j$.
\end{thm}

\begin{cor} \label{cor:GeneralIIDEntropy} \cite{OurIIDEntropy}
Let $X_1,\dots, X_N$ be a collection of i.i.d. random variables taking values in $\chi=\{x_1,\dots,x_d\}$. 
Then the following bounds hold
 \begin{align}
           S\left(\sum_{i=1}^N X_i \right) \geq \frac{r(\chi)}{2} \log( 2\pi e \lambda_1 N )-o(1)
    \end{align}
where $r(\chi)$ is the incommensurability rank of $X$ and $\lambda_1$ is given in Eq.~\eqref{eq: lambda1} for any non-degenerate $r(\chi)$-canonical prepartition of $\chi$.
\end{cor}

\end{widetext}

\end{document}